\documentclass[numbook,envcountsect,envcountsame,envcountreset,smallextended]{svjour3}
\usepackage{mathptmx}
\smartqed

\usepackage[round]{natbib}

\usepackage{paralist}
\usepackage{color}
\usepackage{amsmath,amssymb,amsfonts}
\usepackage{cases}

\newcommand{\nada}[1]{}
\nada{
\newtheorem{theorem}{Theorem}
\newtheorem{lemma}{Lemma}

}

\newcommand{\mux}{\mu^{X}}
\newcommand{\sigmax}{\sigma^{X}}
\newcommand{\muf}{\mu^{F}}
\newcommand{\sigmaf}{\sigma^{F}}
\newcommand{\SX}{\nu^{X}}
\newcommand{\SF}{\nu^{F}}
\newcommand{\pix}{\pi^{X}}
\newcommand{\pif}{\pi^{F}}
\newcommand{\optpix}{\hat{\pi}^{X}}
\newcommand{\optpif}{\hat{\pi}^{F}}
\DeclareMathOperator{\ESR}{ESR}

\title{Hedge and Mutual Funds' Fees \\ and the Separation of Private Investments}

\author{Paolo Guasoni\thanks{Paolo Guasoni is partially supported by the ERC (278295), NSF (DMS-1109047, DMS-1412529), SFI (07/SK/M1189, 08/SRC/FMC1389), and the European Commission (RG-248896).}
\and
Gu Wang
}

\institute{Paolo Guasoni\\
Boston University,
Department of Mathematics and Statistics,
111 Cummington Mall,
Boston, MA 02215,
United States
\and
Dublin City University,
School of Mathematical Sciences,
Glasnevin,
Dublin 9,
Ireland\\
\email{guasoni@bu.edu}\\
Gu Wang\\
University of Michigan at Ann Arbor,
Department of Mathematics,
530 Church Street,
Ann Arbor, MI 48109,
United States\\
\email{robuw@umich.edu}
}

\date{\today}

\begin{document}
\maketitle

\begin{abstract}
A fund manager invests both the fund's assets and own private wealth in separate but potentially correlated risky assets, aiming to maximize expected utility from private wealth in the long run. If relative risk aversion and investment opportunities are constant, we find that the fund's portfolio depends only on the fund's investment opportunities, and the private portfolio only on private opportunities. This conclusion is valid both for a hedge fund manager, who is paid performance fees with a high-water mark provision, and for a mutual fund manager, who is paid management fees proportional to the fund's assets.
The manager invests earned fees in the safe asset, allocating remaining private wealth in a constant-proportion portfolio, while the fund is managed as another constant-proportion portfolio. The optimal welfare is the maximum between the optimal welfare of each investment opportunity, with no diversification gain.
In particular, the manager does not use private investments to hedge future income from fees.
\end{abstract}




\section{Introduction}
Performance fees are the main driver of hedge fund managers' compensation. Typical performance fees amount to 20\% of a fund's profits, subject to a high-water mark provision, which requires past losses to be recovered before further fees are paid.  By contrast, mutual fund managers receive as management fees a fixed proportion -- about 1\% annually -- of assets under management.

A manager's large exposure to the fund's performance is a powerful incentive to deliver superior returns, but is also a potential source of moral hazard, as the manager may use private wealth to hedge such exposure. Extant models \citep{ross2004compensation,carpenter2000does,panageas2009hwm,guasoni2012nopart} acknowledge this issue, but avoid it rather than modeling it, by assuming that private wealth, including earned fees, are invested at the risk-free rate. As a result, the literature is virtually silent on the interplay between a manager's personal and professional investments.\footnote{As an exception, \citet*{aragon2010} restrict earned fees to reinvestment in the fund, which also excludes hedging attempts.}

This paper begins to fill this gap, focusing on a model with two investment opportunities, one accessible to the fund, the other accessible to the manager's private account. Investment opportunities are constant over time, and potentially correlated. To make the model tractable, and consistently with the literature, we consider a fund manager with a constant relative risk aversion and a long horizon, who maximizes expected utility from private wealth. The assumption of a long horizon means, in particular, that the model's conclusions are driven by a stationary risk-return tradeoff, and are indifferent to the one-time, short-term profits that materialize near the horizon, such as by taking excessive risk in the fund, in order to increase the probability of earning more fees before the horizon approaches.

For both hedge and mutual fund managers, we find the manager's optimal investment policies explicitly. The optimal portfolio for the fund entails a constant risky proportion, which corresponds to the manager's own risk aversion for mutual funds, while for hedge funds it corresponds to the effective risk aversion identified by \citet*{guasoni2012nopart} in the absence of private investments.

The optimal policy for private wealth is more complex. For both hedge and mutual funds, the manager leaves earned fees in the safe asset, investing remaining wealth according to an optimal constant-proportion portfolio, which corresponds to the manager's own risk aversion. The result of these policies combined is that the manager obtains the maximum between fees' and private investments' welfare, but not more.

The significance of this result is threefold.
First, the model predicts that the fund composition does not affect the manager's private investments, and that such investments also do not affect the fund composition -- portfolio separation holds. In particular, even if investment opportunities are positively correlated, the manager does not attempt to hedge fund exposure with a position in the private account. The intuition is that, for a long horizon, the benefits from hedging are surpassed by the costs of holding a short position in an asset with positive returns.

Second, the manager does not rebalance all private wealth. Indeed, the optimal policy is to leave earned fees in the safe asset, and to rebalance only excess wealth. This policy effectively replicates a pocket of private wealth that grows like the fund (for mutual funds) or its high-water mark (for hedge funds), while leaving the other pocket to grow at the optimal rate for private investments. Over time, the pocket with the higher growth rate dominates the private portfolio, delivering the maximum welfare of the two strategies. In contrast to usual portfolio allocation with multiple assets, private investments can outperform the fund, but cannot augment its Sharpe ratio through diversification, regardless of correlation.

Third, since the manager's welfare is the maximum between the fund's and the private wealth's welfares, our policy is always optimal, but never unique. Indeed, if the fund delivers the optimal welfare, it does not matter how the manager invests private wealth in excess of earned fees. By contrast, if private investments deliver the optimal welfare, it does not matter how earned fees or even the fund are invested. Although the lack of uniqueness is an extreme effect of the long-horizon approximation, it highlights that either the fund, or private investments, become the main focus of a manager, without long-lasting interactions. Furthermore, the model yields the conditions under which a manager focuses on the fund rather than on private wealth.

In summary, for a manager with a long horizon we find that neither performance fees nor management fees create the incentive to use private investments to either hedge or augment the fund's returns. This conclusion remains valid if the manager has private access to the fund's investment opportunities, a situation that is nested in our model when the correlation between investment opportunities is perfect, and the Sharpe ratios are equal.

The results can inform the decisions of investors and regulators alike. For investors, the main message is that moral hazard can be mitigated by arrangements that increase a manager's horizon, such as longer lock-up periods, infrequent redemptions, deferred compensation, or clawback provisions. These observations are broadly consistent with those of \citet*{aragon2010}, who find that some of these features help alleviate asymmetric-information issues for hedge funds. From a regulatory viewpoint, our results suggest that, in order to reduce managers' incentive to privately trade against investors' interest, restrictions on managers' private investments may be less important than incentive contracts that encourage managers to plan with a long term perspective.

The analysis of this paper relies on the maximization of the \emph{equivalent safe rate} for constant relative risk aversion and a long horizon, which summarizes the expected utility of the management contract in terms of an equivalent risk-free rate earned on private wealth in lieu of the risky flow of fees and returns. This criterion focuses on the stationary, persistent tradeoffs of the model, which dominate its long-term performance, while neglecting initial conditions, such as the initial ratio between the fund's assets and the manager's private wealth. 

Importantly, this asymptotic criterion is immune to the well-known fallacies of growth-optimal portfolios \citep{merton1974fallacy,samuelson1974generalized}, in which the first two asymptotic moments of logarithmic wealth are used to approximate expected utility, because the equivalent safe rate measures the exponential rate of growth of \emph{true} expected utility with the horizon. Rather, the limit of this criterion lies in the relative importance of such a rate in determining expected utility, equivalent to the question of whether the horizon is long enough -- which is not addressed in this paper.

The effectiveness of the equivalent safe rate is thus a quantitative (ultimately, empirical) issue,  similar to the relevance of turnpike portfolios \citep{leland1972turnpike} for long horizons, which depends critically on the parameter values considered \citep{dybvig1999portfolio}. 
The equivalent safe rate performs well with horizons as short as two years in portfolio choice with transaction costs \citep{gerhold2014transaction}, and with horizons of about twenty years with stochastic investment opportunities \citep{guasoni2012portfolios}, while estimates for the problem considered here remain an open question.


Such a word of caution notwithstanding, the main message of the paper -- that long horizons mitigate moral hazards for fund managers -- can be read in reverse as pointing to short horizons as an essential ingredient of any model of portfolio delegation that reproduces substantial moral hazards.

The literature on turnpike portfolios also suggests that our results are likely to remain valid in the same form for utility functions that are only approximately isoelastic for high levels of wealth. Such a generalization is not pursued here, but its conclusions would require the double qualification that the horizon be long enough so that both initial conditions are irrelevant, \emph{and} wealth grows fast enough into the domain of approximate isoelasticity. On the other hand, relaxing the assumption made here of constant investment opportunities is presumably more challenging. At best, the parametric conditions for attention separation would be highly model-dependent. At worst, the results may simply fail to hold, for example, in (admittedly improbable) models where personal and professional investment opportunities alternately exclude each other, therefore the manager can achieve a superior return by alternating attention to whichever opportunity is present.

The paper is most closely related to the literature in portfolio choice with hedge funds \citep{detemple2010optimal}, high-water marks \citep{panageas2009hwm,janecek2012optimal}, and drawdown constraints \citep{grossman1993ois,cvitanic1995portfolio,elie2008optimal}, and is the first one to consider a manager who simultaneously trades in the fund and in private wealth to maximize personal welfare. The rest of the paper is organized as follows: the next section presents the model, its solution, and discusses the main implications. Section 3 offers a heuristic derivation of the main result using informal arguments of stochastic control, and Section 4 concludes. The formal verification of the main theorems are in the Appendix.

\section{Main Result}

\subsection{Model}
A fund manager aims at maximizing expected utility from private wealth at a long horizon. (For brevity, henceforth `private wealth' is simply `wealth', unless ambiguity arises.) To achieve this goal, the manager has two tools: allocating the fund's assets $X$ between a safe asset and a risky asset $S^X$, and allocating wealth $F$, including fees earned from the fund, between the safe asset and another risky asset $S^F$.
The interpretation is that the manager's wealth is invested in securities available to individual investors, while the fund may have access to investment opportunities that, because of scale, regulation, or technology, are restricted to institutional investors. Examples of such investments are institutional funds, restricted shares, such as Rule 144a securities, or high-frequency trading strategies.

\subsubsection*{Market}

The fund's ($S^X$) and private ($S^F$) risky assets follow two correlated geometric Brownian Motions, with expected returns, volatilities and Sharpe ratios $\mux, \sigmax$, $\SX = \frac{\mux}{\sigmax}$ and $\muf, \sigmaf$, $\SF = \frac{\muf}{\sigmaf}$ respectively. Formally, consider a filtered probability space $(\Omega,\mathcal F,(\mathcal F_t)_{t\ge 0},P)$ equipped with the Brownian Motions $W^X = (W^X_t)_{t\ge 0}$ and $W^F = (W^F_t)_{t\ge 0}$, with correlation $\rho$ (i.e., $\langle W^X, W^F\rangle_t =\rho t$), and define the risky assets as
\begin{align*}
\frac{dS^{X}_{t}}{S^{X}_{t}} =& \mux dt + \sigmax dW^{X}_{t},\\
\frac{dS^{F}_{t}}{S^{F}_{t}} =& \muf dt + \sigmaf dW^{F}_{t}.
\end{align*}
To ease notation, we assume a zero safe rate.\footnote{\citet*{guasoni2012nopart} consider a constant safe rate, and find that its value does not affect the optimal fund's policy, suggesting that the assumption of a zero safe rate is inconsequential.}
The manager chooses the proportion of the fund $\pix = \left(\pix_t\right)_{t\geq 0}$ to invest in the asset $S^X$, and the proportion of wealth $\pif = \left(\pif_{t}\right)_{t\geq 0}$ to invest in the asset $S^F$. The strategies $\pix$ and $\pif$ are square-integrable processes, adapted to $\mathcal F_t$, defined as the augmented natural filtration of $W^X$ and $W^F$.

\subsubsection*{Hedge Funds}

In a hedge fund, the manager receives a flow of performance fees subject to a high-water mark provision. The high-water mark $\bar X_{t}$ is the running maximum $\bar X_{t} = \displaystyle\max_{0\leq s\leq t}X_{s}$ of the net value of the fund.\footnote{In the rest of the paper, for any process $(X_t)_{t\geq0}$, $\bar X_t$ denotes its running maximum.}  Then, with the investment strategy $\pix$, the net fund return equals the gross return on the amount invested $X^{\pix}_{t}\pix_{t}$, minus performance fees, which are a fraction of the increase in the high-water mark. Thus, with $0<\alpha<1$,
\begin{equation}
dX^{\pix}_{t} = X^{\pix}_{t}\pix_{t}\left(\mux dt + \sigmax dW^{X}_{t}\right) - \frac{\alpha}{1-\alpha}d\bar X^{\pix}_{t}.\label{DymX}
\end{equation}
In this equation, the last term reflects the fact that each dollar of gross profit is split into $\alpha$ dollars as fees, plus $1-\alpha$ dollars as net profit, whence performance fees are $\alpha/(1-\alpha)$ times net profit.

Similarly, the return on the manager's wealth equals the return on the risky wealth $F^{\pix,\pif}_{t}\pif_{t}$, plus the fees earned from the fund, i.e.\begin{equation}
dF^{\pix,\pif}_{t} = F^{\pix,\pif}_{t}\pif_{t}\left(\muf dt + \sigmaf dW^{F}_{t}\right) + \frac{\alpha}{1-\alpha}d\bar X^{\pix}_{t}.\label{DymF}
\end{equation}
Note, in particular, that while the fund evolution depends only on its policy $\pix$, the evolution of wealth $F^{\pix,\pif}_{t}$ depends both on $\pif$ and on $\pix$, as the latter drives earned fees.

\subsubsection*{Mutual Funds}

For a mutual fund manager, fees are proportional to assets under management, and $\varphi>0$ denotes the annual fees as a fraction of assets.
Henceforth, with a slight abuse of notation we use the letters $X$ and $F$ to denote the fund and wealth processes both in the model for hedge funds and in the model for mutual funds. This notation has the advantage of emphasizing the analogies between the two models, and should not create confusion, since hedge and mutual fund managers exclude each other.

Thus, in the mutual fund model the joint dynamics of fund and wealth becomes

\begin{align}
dX^{\pix}_{t} =& X^{\pix}_{t}\pix_{t}\left(\mux dt + \sigmax dW^{X}_{t}\right) -\varphi X^{\pix}_tdt,\label{DymXM}\\
dF^{\pix,\pif}_{t} =& F^{\pix,\pif}_{t}\pif_{t}\left(\muf dt + \sigmaf dW^{F}_{t}\right) + \varphi X^{\pix}_tdt\label{DymFM}.
\end{align}
whence the return on the fund $d X^{\pix}_{t}/X^{\pix}_{t}$ is decreased by the constant $\varphi$, while the return on wealth $dF^{\pix,\pif}_{t}/F^{\pix,\pif}_{t}$ is increased by a variable amount, which depends on the fund/wealth ratio $X^{\pix}_t/F^{\pix,\pif}_t$.

\subsubsection*{Preferences}

In both cases, the fund manager chooses optimal strategies $\optpix$ and $\optpif$ so as to maximize expected utility from wealth in the long run, that is, the equivalent safe rate (ESR)  of  wealth (cf. \citet*{grossman1993ois,dumas1991esd,cvitanic1995portfolio}):
\begin{equation}
\ESR_\gamma({\pix,\pif}) =
\begin{cases}
\displaystyle\lim_{T\rightarrow \infty}\frac{1}{T} \ln\mathbb{E}\left[\left(F^{\pix,\pif}_{T}\right)^{1-\gamma}\right]^{\frac{1}{1-\gamma}},  & 0< \gamma \neq 1,\\
\displaystyle\lim_{T\rightarrow \infty} \frac{1}{T} \mathbb{E}\left[\ln F^{\pix,\pif}_{T}\right],\label{logaim}
& \gamma=1.
\end{cases}
\end{equation}
This equivalent safe rate measures the manager's welfare, and has the dimension of an interest rate. It corresponds to the hypothetical safe rate which would make the manager indifferent between (i) actively managing the fund and wealth, and (ii) retiring from the fund, investing all wealth at this riskless rate.

\subsection{Solution and Discussion}
The main result identifies the manager's optimal policies, and the corresponding welfare. The result below is proved in the case of logarithmic risk aversion ($\gamma=1$), or lower ($\gamma\le 1$), and hence allows to understand the risk-neutral limit $\gamma\downarrow 0$. We conjecture that the same result remains valid for risk aversion greater than one, but we cannot offer a formal proof for this case.

In the next two theorems, $\gamma\in (0,1]$ denotes the manager's relative risk aversion.
\begin{theorem}[Hedge Funds]\label{Main}
For a hedge fund manager compensated by high-water mark  performance fees with rate $0<\alpha<1$, the investment policies
\begin{align}
\optpix_{t} &= \frac{\mux}{\gamma^{*}\left(\sigmax\right)^{2}},
\quad\text{where}\quad \gamma^{*} = \alpha + (1-\alpha)\gamma\label{optimal1}\\
\optpif_{t} &= \left(1-\frac{\alpha}{1-\alpha}\frac{\bar{\hat{X}}_{t} - X_{0}}{\hat{F}_{t}}\right)\frac{\muf}{\gamma\left(\sigmaf\right)^{2}},\label{optimal2}
\end{align}
attain the manager's maximum equivalent safe rate of  wealth, which equals
\begin{equation}
\ESR_\gamma({\optpix,\optpif}) =
\max\left((1-\alpha)\frac{\left(\SX\right)^{2}}{2 \gamma^*},\frac{\left(\SF\right)^{2}}{2\gamma}\right),\label{ESR}
\end{equation}
where $\hat{X}$ and $\hat{F}$ are the fund and wealth processes corresponding to $\optpix$ and $\optpif$.
\end{theorem}

\begin{theorem}[Mutual Funds]\label{Mutual}
For a mutual fund manager compensated by proportional fees with rate $\varphi>0$, the investment policies
\begin{align}
\optpix_{t} &= \frac{\mux}{\gamma\left(\sigmax\right)^{2}},\label{optimal1M}\\
\optpif_{t} &= \left(1-\frac{\varphi\int_0^t\hat{X}_{s}ds}{\hat{F}_{t}}\right)\frac{\muf}{\gamma\left(\sigmaf\right)^{2}},\label{optimal2M}
\end{align}
attain the manager's maximum equivalent safe rate of wealth, which equals
\begin{equation}
\ESR_\gamma({\optpix,\optpif}) =
\max\left(\frac{\left(\SX\right)^{2}}{2\gamma} - \varphi,\frac{\left(\SF\right)^{2}}{2\gamma}\right),\label{ESRM}
\end{equation}
where $\hat{X}$ and $\hat{F}$ are the fund and wealth processes corresponding to $\optpix$ and $\optpif$.
\end{theorem}

Both fund investment policies \eqref{optimal1} and \eqref{optimal1M} imply constant-proportion portfolios, but with different exposures. For a hedge fund manager, the optimal policy corresponds to the risk aversion $\gamma^*$, which is between one and the manager's own risk aversion $\gamma$, and depends on the rate of performance fees $\alpha$. This policy coincides with the one obtained by \citet*{guasoni2012nopart} in the absence of private investment opportunities, which is recovered in our model with $\nu^F=0$. In this case, the private risky opportunity has zero return, and hence it is never used. By contrast, for a mutual fund manager the optimal policy corresponds to the manager's own risk aversion $\gamma$, and does not depend on the rate of proportional fees $\varphi$.

The private policies in \eqref{optimal2} and \eqref{optimal2M} are extremely similar, and are best understood by denoting cumulative fees up to time $t$ by $C_t$, whereby $C_t = \frac{\alpha}{1-\alpha}\left(\bar{\hat{X}}_{t} - X_{0}\right)$ for a hedge fund, while $C_t = \varphi\int_0^t\hat{X}_{s}ds$ for a mutual fund. With this notation, in each model the total risky and safe positions respectively reduce to
\begin{align*}
\hat{F}_{t}\optpif_t &=
\left(\hat{F}_{t}-C_t\right)\frac{\muf}{\gamma\left(\sigmaf\right)^{2}},\\
\hat{F}_{t} (1-\optpif_{t}) &=
\left(\hat{F}_{t}-C_t\right)\left(1-\frac{\muf}{\gamma\left(\sigmaf\right)^{2}}\right)+ C_t .
\end{align*}
These formulas show that for both hedge and mutual funds, the manager divides wealth into earned fees $C_t$, which are set aside in the safe asset, and the rest, which is invested in the constant-proportion portfolio with the manager's own risk aversion $\gamma$. With this investment strategy $\hat{F}_t = F_0e^{\frac{(2\gamma-1)\left(\SF\right)^2t}{2\gamma^2} + \frac{\muf }{\gamma\sigmaf}W^{F}_t} + C_t$, and the managers' welfare in \eqref{ESR} and \eqref{ESRM} equal the maximum between the welfare of fees and the welfare of private investments.

\subsection{Portfolio Separation}

A salient feature of these results is that the risky position in the fund policy is independent of private positions, and vice versa. In other words, the risky investment in the fund $\hat{X}_t\optpix_t$ does not depend on $\muf, \sigmaf$, and the risky investment in wealth $\hat{F}_t\optpif_t = F_0e^{\frac{(2\gamma-1)\left(\SF\right)^2t}{2\gamma^2}+ \frac{\muf W^{F}_t}{\gamma\sigmaf}}\frac{\muf}{\gamma\left(\sigmaf\right)^2}$ does not depend on $\mux, \sigmax$. Furthermore, neither $\optpix$ nor $\optpif$ depend on the correlation $\rho$ between investment opportunities. We call this property portfolio separation.\footnote{Note that the \emph{proportion} of risky private investment depends on total private wealth, hence on cumulative fees $C_t$, which in turn depend on the fund proportion $\optpix_t$. Thus, portfolio separation also holds for proportions, but only \emph{conditionally} on the values of the fund and private accounts. By contrast, the dollar value of the private risky investment $\hat{F}_t\optpif_t$ does not depend on cumulative fees (which are left in the safe investment), and therefore portfolio separation holds even \emph{unconditionally} for dollar positions.}

Portfolio separation entails that a manager with a long horizon has no incentive to hedge personal exposure to future fee income with private risky assets, regardless of their correlation. In fact, hedging does not take place even in the limit case $\mux=\muf, \sigmax=\sigmaf, \rho=1$, which corresponds to a manager who has unfettered access to the fund's investments with private capital, and hence faces a dynamically complete market. This result holds for both high-water mark performance fees, and proportional management fees.

The biggest concern about high-water mark performance fees combined with private investment opportunities is that when the fund and private investment opportunities are highly correlated, managers may short private investments and take a larger risk in the fund, which may be against clients' interests. To understand why such hedging is ineffective, suppose that risky assets are positively correlated, and consider a manager with a fund trading well below its high-water mark. In this case, a short position in the private asset is a poor hedge, because the high-water mark (and hence future income) is insensitive to small variations in the fund value, while the short position reduces the long-term growth of wealth.

Consider now a mutual fund manager. If the two investment opportunities are the same, the manager could hedge the flow of future fees with a strategy that swaps their unique arbitrage-free price with a flow of payments that exactly offset the flow of fees. Such a strategy is equivalent to selling future fees forward, and therefore it entails a short position in the risky asset, leaving the manager with a zero net risky exposure. However, such an exposure is not optimal, because it foregoes the growth in wealth generated by that risky asset, and thus the manager wants to keep a positive allocation to the risky asset, which reverses at least in part, and possibly more, the hedging position.

The mathematical point is that the fee flow is a finite variation process both for the hedge fund manager, since the high-water mark is increasing, and for the mutual fund manager, since cumulative fees are absolutely continuous in time. Thus, though the two investment opportunities may be correlated, the fee flow has zero covariation with the private investment.

Put another way, portfolio separation implies that the manager has no incentive to take more or less risk in the fund, in view of private investment opportunities outside the fund. A priori, it may seem plausible that a manager takes more risk in the fund if outside opportunities are attractive, because more risk is likely to lead to more fees in the short term, which could then be invested in outside opportunities. However, this tactic would only generate a one-time wealth transfer, not a lasting increase in the growth rate of manager's wealth, and hence is irrelevant in the long run.

In summary, the message of portfolio separation is largely positive: if horizons are long, then moral hazards concerns are limited, because performance or proportional fees essentially defeat any hedging incentives between fund and wealth. Yet, portfolio separation has a downside -- attention separation.

\subsection{Attention Separation}

As a consequence of portfolio separation, the manager's welfare in \eqref{ESR} and \eqref{ESRM} is the maximum between the welfare from fees and the welfare from remaining wealth. Thus, while the joint policies in \eqref{optimal1}, \eqref{optimal2} and \eqref{optimal1M}, \eqref{optimal2M} are optimal in all cases, they are never unique. Indeed, if the manager's welfare is due to the fund (that is, fees), then the private investment opportunity  becomes irrelevant, and $\optpif$ can be replaced, for example, with the policy $\pif=0$. Vice versa, if remaining wealth drives the welfare, then the fund policy is irrelevant, and utter negligence ($\pix=0$) will deliver the same result.

This rather extreme implication is driven by the assumption of a long-horizon, which focuses on the risk-adjusted long-term growth rate, neglecting all short-term effects. Still, it makes it clear that a manager's commitment to the fund will easily wane, unless its investments are superior to outside opportunities. The manager's attention inevitably shifts to either the fund, or wealth, whichever is more profitable.

For a hedge fund, equation \eqref{ESR} shows that the manager focuses on the fund if and only if the fund's Sharpe ratio $\SX$ exceeds the private Sharpe ratio $\SF$ by a multiple, which depends on the fund's rate of performance fees and on the manager's risk aversion
\begin{equation}
\frac{\SX}{\SF} > \left(1+\frac{\alpha}{(1-\alpha)\gamma}\right)^{\frac{1}{2}}.
\label{eq:minatt}
\end{equation}
For example, in the case of a logarithmic manager $\gamma=1$, and of performance fees of $20\%$, the manager focuses on the fund, provided that its Sharpe ratio is 11.8\% higher than that of private investments. Such a condition is likely to hold in practice: \citet*{getmansky2004econometric} find high Sharpe ratios in the hedge fund industry, even after controlling for return smoothing and illiquidity.

The right-hand side in \eqref{eq:minatt}, which represents the manager's attention threshold, grows as risk aversion declines. The explanation is as follows: as risk aversion declines to zero, the effective risk aversion $\gamma^*=\alpha+(1-\alpha)\gamma$ induced by the high-water mark converges to $\alpha$, which entails finite leverage in the fund. On the other hand, the private portfolio is driven by the true risk aversion $\gamma$, which declines to zero, leading to increasingly high leverage. Because leverage can arbitrarily  magnify expected returns, for sufficiently low risk aversion the private portfolio is always more attractive.

Analagously, for a mutual fund, the manager focuses on the fund under the condition $\frac{(\SX)^2}2-\frac{(\SF)^2}2 > \gamma \varphi$, which means that the squared Sharpe ratio of the fund must exceed its private counterpart by twice the product between the risk aversion and the rate of management fees. In other words, because proportional fees do not affect the fund investment policy, but merely reduce its return, the condition is that the squared Sharpe ratio of the fund must exceed the one of wealth enough to  compensate for the loss due to fees. A lower risk aversion means that the fund has a higher exposure to the risky asset. Because the fund's equivalent safe rate before fees $\frac{(\SX)^2}{2\gamma}$ increases as $\gamma$ decreases, while the proportional fee $\varphi$ remains fixed, $\varphi$ becomes a smaller fraction of total welfare $\frac{(\SX)^2}{2\gamma}-\varphi$, and therefore it is less important for an aggressive manager than it is for a conservative one. As long as the fund is more attractive than private investments $\SX>\SF$, a sufficiently aggressive manager will focus on the fund even if the fee is high, because its relative impact is small.

Overall, attention separation brings both some bad news, as the manager may grossly neglect the fund if it does not offer sufficiently attractive returns, and some good news, since the conditions for attention to the fund seem mild, and, in the case of hedge funds, a manager with low risk aversion is likely to leverage wealth rather than the fund.

\subsection{Growth and Fees}

A puzzling feature of extant models of performance fees is that a manager prefers lower performance fees, i.e. welfare is decreasing in $\alpha$. The explanation of this finding, common to the models of \citet*{panageas2009hwm} with risk-neutrality, and of \citet*{guasoni2012nopart} with risk aversion, is that higher fees today reduce the growth rate of the fund, leading to lower fees tomorrow. Both models assume that fees are invested at the safe rate in the manager's account, and raise the question of whether reinvestment can induce the preference for higher fees.

Equation \eqref{ESR} offers a qualified negative answer. If private risky investments are available, there will be some threshold $\alpha^*$, below which the manager prefers lower fees, as in the absence of private investments, and above which the manager is indifferent to changes in fees, because the fund becomes irrelevant, as the welfare is entirely driven by wealth. This threshold is in fact the value of $\alpha$ for which \eqref{eq:minatt} holds as equality.

This result is essentially a consequence of portfolio separation. Because the manager is unable to compound the fund growth with wealth growth, either private investments make fees negligible, or are negligible themselves.
Overall, the model shows that the reinvestment value of fees is not sufficient to obtain the manager's preference for higher payout rates, which in turn is likely to involve intertemporal consumption or fund flows.

\section{Heuristic Solution}

This section derives a candidate optimal solution for maximizing the ESR of a hedge fund manager's wealth, who is compensated by high-water mark performance fees, with heuristic stochastic control arguments. For brevity, this argument is presented only for the case of logarithmic utility, while the rigorous proof for all cases  $0 < \gamma \leq 1$ is in the Appendix.

To ease notation, in the rest of the paper we drop the superscripts $\pix$ and $\pif$ from $X_{t}$ and $F_{t}$, if no ambiguity arises.
Denoting the state variable $\bar X_{t} - X_{0}$ by $Z_{t}$, the manager's value function of utility maximization from terminal wealth is
\begin{equation*}
V(t,x,f,z) = \sup_{\pix,\pif}\mathbb{E}_{t}\left[\ln F_{T} | X_{t} = x, F_{t} = f, Z_{t} = z\right].
\end{equation*}
For $V$ as a stochastic process (because it is a function of $X_t$, $F_t$ and $Z_t$), by I\^to's formula,
\begin{align*}
dV&= \sigmax x \pix_{t}\frac{\partial V}{\partial x}dW^{X}_{t} + \sigmaf f\pif_{t}\frac{\partial V}{\partial f}dW^{F}_{t}+ \left(\frac{\partial V}{\partial z}+ \frac{\alpha}{1-\alpha}\left(\frac{\partial V}{\partial f} -
\frac{\partial V}{\partial x}\right)\right)d\bar X_{t}\\
+&\left(\frac{\partial V}{\partial t} + x\mux\pix_{t} \frac{\partial V}{\partial x} + f\muf\pif_{t}
\frac{\partial V}{\partial f}\right)dt\\
 +& \left(\frac{\left(\sigmax x \pix_{t}\right)^{2}}{2}\frac{\partial^2V}{\partial x^2} +
\frac{\left(\sigmaf f \pif_{t}\right)^{2}}{2}\frac{\partial^2V}{\partial f^2} + \rho \sigmax \sigmaf xf\pix_{t} \pif_{t}
\frac{\partial^2V}{\partial x\partial f}\right)dt.
\end{align*}
\normalsize
Thus the Hamilton-Jacobi-Bellman (HJB) equation for $V(t,x,f,z)$ is, for $0<x<z + X_{0}$, (where $X_{0}$ is the initial fund's value)
\small
\begin{equation*}
\textstyle
-\frac{\partial V}{\partial t} = \sup\limits_{\pix,\pif}\left(x\mux\pix \frac{\partial V}{\partial x} + f\muf\pif \frac{\partial V}{\partial f} +
\frac{\left(\sigmax x \pix\right)^{2}}{2}\frac{\partial^2 V}{\partial x^2} +
\frac{\left(\sigmaf f \pif\right)^{2}}{2}\frac{\partial^2V}{\partial f^2} + \rho \sigmax \sigmaf xf\pix \pif
\frac{\partial^2V}{\partial x\partial f}\right),
\end{equation*}
\normalsize
with the boundary condition:
\begin{equation*}
\frac{\partial V}{\partial z} + \frac{\alpha}{1-\alpha}(\frac{\partial V}{\partial f} - \frac{\partial V}{\partial x}) = 0  \text{ when }  x = z + X_{0}.
\end{equation*}
Note that here and in the rest of this section we use $\pix$ and $\pif$ (previously used to denote investment strategies as stochastic processes) to denote feedback controls, which are functions of state variables, when no ambiguity arises.

By the usual scaling property of logarithmic utility, and in the long-horizon limit, we can rewrite the value function as $V(t,x,f,z) = -\beta t + \ln z + v(\xi,\phi)$, where
$\xi = \ln\frac{x}{z}$ and $\phi = \ln\frac{f}{z}$. Then in terms of $v$, the HJB equation becomes
\small
\begin{equation*}
\textstyle
\beta = \sup\limits_{\pix,\pif}\left(\mux\pix \frac{\partial v}{\partial \xi} + \muf\pif \frac{\partial v}{\partial \phi} +
\frac{\left(\sigmax\pix\right)^{2}}{2}(\frac{\partial^2v}{\partial\xi^2}-\frac{\partial v}{\partial \xi}) +
\frac{\left(\sigmaf\pif\right)^{2}}{2}(\frac{\partial^2v}{\partial \phi^2}-\frac{\partial v}{\partial\phi}) + \rho \sigmax
\sigmaf\pix\pif \frac{\partial^2v}{\partial\xi\partial\phi}\right),
\end{equation*}
\normalsize
for $ 0<x<z+X_{0}$, while the boundary condition reduces to
\begin{equation*}
(1-\alpha) - \frac{\partial v}{\partial \xi}\left(\alpha \exp(-\xi) + (1-\alpha)\right) +
\frac{\partial v}{\partial \phi}\left(\alpha \exp(-\phi) - (1-\alpha)\right) = 0,
\end{equation*}
when $x = z + X_{0}$.

In the long run, the initial fund's value $X_{0}$ should not matter in this optimization problem. Furthermore, since $\bar X_{t}$ becomes large in the long run at optimum, $\bar X_{t}\approx Z_{t} = \bar X_{t} - X_{0}$, and we can approximate the HJB equation and the boundary condition with
\small
\begin{equation}
\textstyle
\beta = \sup\limits_{\pix,\pif}\left(\mux\pix \frac{\partial v}{\partial \xi} + \muf\pif \frac{\partial v}{\partial \phi} +
\frac{\left(\sigmax\pix\right)^{2}}{2}(\frac{\partial^2v}{\partial\xi^2}-\frac{\partial v}{\partial \xi}) +
\frac{\left(\sigmaf\pif\right)^{2}}{2}(\frac{\partial^2v}{\partial \phi^2}-\frac{\partial v}{\partial\phi}) + \rho \sigmax
\sigmaf\pix\pif \frac{\partial^2v}{\partial\xi\partial\phi}\right) , \label{HJB}
\end{equation}
\normalsize
when $-\infty<\xi<0$, and
\begin{equation}
(1-\alpha) - \frac{\partial v}{\partial \xi} +
\frac{\partial v}{\partial \phi}\left(\alpha \exp(-\phi) - (1-\alpha)\right) = 0,\label{BC}
\end{equation}
when $\xi = 0$.

The first order conditions for $\pix$ and $\pif$ in  (\ref{HJB}) are, as functions of $\xi$ and $\phi$,
\begin{align}
\optpix(\xi,\phi) &= -\frac{\sigmaf\mux
\frac{\partial v}{\partial\xi}\left(\frac{\partial^2v}{\partial\phi^2}-\frac{\partial v}{\partial\phi}\right)-\rho\sigmax \muf
\frac{\partial^2v}{\partial\xi\partial\phi}\frac{\partial v}{\partial\phi}}{\left(\sigmax\right)^{2}\sigmaf\left(\left(\frac{\partial^2v}{\partial\xi^2}-\frac{\partial v}{\partial\xi}\right)\left(\frac{\partial^2v}{\partial\phi^2}-\frac{\partial v}{\partial\phi}\right)-\left(\rho\frac{\partial^2v}{\partial\xi\partial\phi}\right)^{2}\right)},\label{policy1}\\
\optpif(\xi,\phi) &=-\frac{\sigmax\muf
\frac{\partial v}{\partial\phi}\left(\frac{\partial^2v}{\partial\xi^2}-\frac{\partial v}{\partial\xi}\right)-\rho\sigmaf \mux
\frac{\partial^2v}{\partial\xi\partial\phi}\frac{\partial v}{\partial\xi}}{\sigmax\left(\sigmaf\right)^{2}\left(\left(\frac{\partial^2v}{\partial\xi^2}-\frac{\partial v}{\partial\xi}\right)\left(\frac{\partial^2v}{\partial\phi^2}-\frac{\partial v}{\partial\phi}\right)-\left(\rho\frac{\partial^2v}{\partial\xi\partial\phi}\right)^{2}\right)}.\label{policy2}
\end{align}
Plugging these two maximizers into (\ref{HJB}) yields
\begin{equation}
\textstyle
\beta =
\frac{\left(\mux\sigmaf\frac{\partial v}{\partial\xi}\right)^{2}\left(\frac{\partial^2v}{\partial\phi^2}-\frac{\partial v}{\partial\phi}\right)-2\rho\sigmax\sigmaf
\muf\mux
\frac{\partial v}{\partial \xi}\frac{\partial v}{\partial \phi}\frac{\partial^2v}{\partial\xi\partial\phi}+\left(\sigmax\muf\frac{\partial v}{\partial\phi}\right)^{2}\left(\frac{\partial^2v}{\partial\xi^2}-\frac{\partial v}{\partial \xi}\right)}{2\left(\sigmax\sigmaf\right)^{2}\left(\left(\frac{\partial^2v}{\partial\xi^2}-\frac{\partial v}{\partial\xi}\right)\left(\frac{\partial^2v}{\partial\phi^2}-\frac{\partial v}{\partial\phi}\right)-\left(\rho\frac{\partial^2v}{\partial\xi\partial\phi}\right)^{2}\right)}, \label{HJB2}
\end{equation}
when $-\infty<\xi<0$.

(\ref{BC}) implies that $\frac{\partial v}{\partial \xi} - \frac{\partial v}{\partial \phi}\left(\alpha \exp(-\phi) - (1-\alpha)\right)$ is a constant for any $\phi \in \mathbb{R}$
when $\xi = 0$. We guess the solution to (\ref{HJB2}) as $v(\xi,\phi) = \delta\xi + b\ln|\alpha - (1-\alpha)\exp(\phi)|$, where $b$ and $\delta$ are parameters to be found.
Plugging this guess into (\ref{BC}) gives $b = 1
- \frac{\delta}{1-\alpha}$. Finally, plugging $v(\xi,\phi) = \delta\xi + \left(1-\frac{\delta}{1-\alpha}\right)\ln|\alpha - (1-\alpha)\exp(\phi)|$ into (\ref{HJB2}),(\ref{policy1}) and (\ref{policy2}) yields
\begin{align}
\beta(\delta) &=\frac{\delta}{2}\left(\SX\right)^{2}+ \left(1-\frac{\delta}{1-\alpha}\right)\frac{1}{2}\left(\SF\right)^{2},\nonumber\\
\optpix (\xi,\phi)&=\frac{\mux}{\left(\sigmax\right)^{2}},\label{logoptimal1}\\
\optpif (\xi,\phi)&= \left(1-\frac{\alpha
z}{(1-\alpha)f}\right)\frac{\muf}{\left(\sigmaf\right)^{2}}.\label{logoptimal2}
\end{align}
Note that $\optpix$ and $\optpif$ do not depend on $\delta$, and taking $\delta$ to be $0$ when
$\frac{1}{2}\left(\SF\right)^{2} \geq \frac{1-\alpha}{2}\left(\SX\right)^{2}$ and $1-\alpha$ when
$\frac{1}{2}\left(\SF\right)^{2} < \frac{1-\alpha}{2}\left(\SX\right)^{2}$ yields (\ref{ESR}). (\ref{logoptimal1}) and (\ref{logoptimal2}) help us conjecture the optimal policies for the fund and the private investments: the manager puts the performance fees, which are $\frac{\alpha z}{(1-\alpha)f}$ of wealth in the safe asset, and invests the rest in a Merton portfolio. For the fund, the same strategy as in \citet*{guasoni2012nopart} is adopted, as to ensure the maximum ESR from performance fees.

\section{Conclusion}

A fund manager with a long horizon does not use private investment opportunities to hedge personal exposure to fund performance. This result holds for a hedge fund manager, paid by performance fees with a high-water mark provision, and for a mutual fund manager, paid by management fees proportional to the fund's assets.
Indeed, optimal  policies for the fund and wealth are separate, in that each of them depends only on the respective investment opportunity. The resulting welfare is the maximum welfare between that of fees and of private investments, without any diversification gain. Thus, the manager effectively focuses on either the fund or wealth, whichever is more attractive.

An important open question, not discussed in this paper, is the quantitative effect of finite horizons. In practice, a very short horizon is likely to lead to substantially different conclusions, as typical investment opportunities offer relatively low returns, but allow an unscrupulous manager to maximize the option-like, one-time payoffs implied by fees by increasing the fund's volatility near the horizon.
The question is to understand whether the short-term or the long-term considerations prevail for medium investment horizons, such as five to ten years, depending on the ratio between the manager's and the fund's assets.

\appendix\normalsize

\section{Verification Theorems}
This appendix contains the proofs of Theorem \ref{Main} and Theorem \ref{Mutual}.
\subsection{Proof of Theorem \ref{Main}}
We start by defining the following processes, which represent cumulative log returns, before fees, on the fund and wealth respectively,
\begin{align*}
R^{X}_{t} &= \int_{0}^{t}\left[\left(\mux\pix_{s}-\frac{1}{2}\left(\sigmax\pix_{s}\right)^{2}\right)ds
+ \sigmax\pix_{s} dW_{s}^{X}\right],\\
R^{F}_{t} &= \int_{0}^{t}\left[\left(\muf \pif_{s}-\frac{1}{2}\left(\sigmaf\pif_{s}\right)^{2}\right)ds
+ \sigmaf\pif_{s} dW_{s}^{F}\right],\\
R^{X}_{t,T} &= R^{X}_{T}-R^{X}_{t},\\
R^{F}_{t,T} &= R^{F}_{T}-R^{F}_{t}.
\end{align*}
$R^{X}$ and $R^{F}$ depend on $\pix$ and $\pif$, respectively, and should be denoted as $R^{X,\pix}$ and $R^{F,\pif}$. We drop the superscripts $\pix$ and $\pif$ for ease of notation, unless it causes ambiguity.

Equations (\ref{DymX}), (\ref{DymF}), Proposition $7$ and Lemma $8$ in \citet*{guasoni2012nopart} imply that
\begin{align}
X_{t} &= X_{0}e^{R^{X}_{t} - \alpha \bar R^{X}_{t}},\nonumber\\
\bar X_{t} &= X_{0}e^{(1-\alpha)\bar R^{X}_{t}},\nonumber\\
F_{t} &=F_{0}e^{R^{F}_{t}} +
\frac{\alpha}{1-\alpha}\int_{0}^{t}e^{R^{F}_{s,t}}d\bar X_{s}\label{F}.
\end{align}

The discussion begins with two simple lemmas that are used often in the proof of the theorems.
\begin{lemma}\label{maxmin}
Let $\left(X_{t}\right)_{t\geq 0}$ and $\left(Y_{t}\right)_{t\geq 0}$ be two continuous stochastic processes, and define $\bar X_{t} = \displaystyle\max_{0\leq s\leq t}X_{s}$ and $\underbar Y_{t} = \displaystyle\min_{0\leq s\leq t}Y_{s}$. Then $\bar X_{t} + \underbar Y_{t} \leq \left(\overline{X+Y}\right)_{t}$.
\end{lemma}

\begin{proof}
Since $\left(\overline{X+Y}\right)_{t} \geq X_{s} + Y_{s} \geq X_{s} + \underbar Y_{t}$ for every $0 \leq s \leq t$, it follows that $\left(\overline{X+Y}\right)_{t}\geq \bar X_{t} + \underbar Y_{t}$.
\end{proof}

\begin{lemma}\label{IntCE}
Let $\left(\mathcal{G}_{t}\right)_{0\leq t\leq T}$ be a continuous filtration, $\mathcal{F} \subseteq \mathcal{G}_{0}$ be a $\sigma$-algebra, and denote by
$\mathbb{E}_{\mathcal{F}}$ and $\mathbb{E}_{\mathcal{G}_{t}}$ the conditional expectation with respect to $\mathcal{F}$ and $\mathcal{G}_{t}$, respectively. If $\left(A_{t}\right)_{t\geq 0}$ is a continuous and
increasing process adapted to $\mathcal{G}_{t}$ for $0\leq t\leq T$, and $\left(X_{t}\right)_{t\geq 0}$ is a positive, continuous stochastic process such that
$\mathbb{E}_{\mathcal{G}_{t}}\left[X_{t,T}\right] \leq C$, for every $0\leq t\leq T$, and some constant $C$, where $X_{t,T} = \frac{X_{T}}{X_{t}}$, then
$\mathbb{E}_{\mathcal{F}}\left[\int_{0}^{T}X_{t,T}dA_{t}\right] \leq C\mathbb{E}_{\mathcal{F}}\left[A_{T}-A_{0}\right]$.
\end{lemma}

\begin{proof}
Since $A_{t}$ is an increasing process, for a partition of $[0,T]$:  $0 = t^{n}_{0} < t^{n}_{1} < \cdots < t^{n}_{n} = T$, $t^{n}_{k}= \frac{k}{n}T$ for $1\leq k \leq n$,
\begin{equation*}
\int_{0}^{T}X_{t,T}dA_{t} = \lim_{n\rightarrow\infty}\sum_{k=1}^{n}X_{t^{n}_{k},T}\left(A_{t^{n}_{k}} - A_{t^{n}_{k-1}}\right).
\end{equation*}
Thus,
\begin{equation*}
\mathbb{E}_{\mathcal{F}}\left[\int_{0}^{T}X_{t,T}dA_{t}\right] = \mathbb{E}_{\mathcal{F}}\left[\lim_{n\rightarrow\infty}\sum_{k=1}^{n}X_{t^{n}_{k},T}\left(A_{t^{n}_{k}} - A_{t^{n}_{k-1}}\right)\right],
\end{equation*}
and by Fatou's Lemma and the tower property of conditional expectation, the right-hand side is less than or equal to
\begin{equation}
\liminf_{n\rightarrow\infty}\mathbb{E}_{\mathcal{F}}\left[\sum_{k=1}^{n}X_{t^{n}_{k},T}\left(A_{t^{n}_{k}} - A_{t^{n}_{k-1}}\right)\right] = \liminf_{n\rightarrow\infty}\mathbb{E}_{\mathcal{F}}\left[\sum_{k=1}^{n}\mathbb{E}_{\mathcal{G}_{t^{n}_{k}}}\left[X_{t^{n}_{k},T}\right]\left(A_{t^{n}_{k}} - A_{t^{n}_{k-1}}\right)\right].\label{sum}
\end{equation}
Since $\mathbb{E}_{\mathcal{G}_{t}}\left[X_{t,T}\right] \leq C$ for every $0\leq t\leq T$, (\ref{sum}) is less than or equal to
\begin{equation*}
C\liminf_{n\rightarrow\infty}\mathbb{E}_{\mathcal{F}}\left[\sum_{k=1}^{n}\left(A_{t^{n}_{k}} - A_{t^{n}_{k-1}}\right)\right] = C\liminf_{n\rightarrow\infty}\mathbb{E}_{\mathcal{F}}\left[A_{T} - A_{0}\right] = C\mathbb{E}_{\mathcal{F}}\left[A_{T} - A_{0}\right].
\end{equation*}
\end{proof}

The proof of Theorem \ref{Main} is divided into the following two steps.
First, any investment policies $\pix$ and $\pif$ satisfy the following:
\begin{equation}
\ESR_{\gamma}(\pi^{X},\pi^{F}) \leq \max\left(\frac{(1-\alpha)\left(\SX\right)^2}{2\gamma^{*}},\frac{\left(\SF\right)^2}{2\gamma}\right),\nonumber
\end{equation}
which is proved in Lemma \ref{less}.
Second, this upper bound is achieved by the candidate optimal policies in (\ref{optimal1}) and (\ref{optimal2}), as proved in Lemma \ref{optimum}.

\begin{lemma}\label{less}
For a hedge fund manager compensated by high-water mark performance fees with rate $0<\alpha<1$, the $\ESR$ induced by any investment strategies $\pix$ and $\pif$ satisfies
\begin{equation*}
\ESR_{\gamma}(\pif,\pif) \leq \lambda_1 = \max\left(\frac{(1-\alpha)\left(\SX\right)^{2}}{2\gamma^{*}},\frac{\left(\SF\right)^{2}}{2\gamma}\right), \text{ for any } 0 < \gamma\leq 1.
\end{equation*}
\end{lemma}

\begin{proof}
We prove this lemma for logarithmic utility and power utility, respectively.\\

\noindent \textit{The case of  logarithmic utility}. For convenience of notation, define
\begin{numcases}
{\tilde{X}_{t}=}
0 & for $t < 0$,\nonumber\\
\frac{1-\alpha}{\alpha}F_{0} + \bar X_{t}-X_{0} & for
$t\geq0$.\nonumber
\end{numcases}
Then $\tilde{X}_{t}$ is an increasing process, which has a jump at $t=0$, and then grows with $\bar X_{t}$.
From (\ref{F}), $F_{t} = \frac{\alpha}{1-\alpha}\int_{0}^{T}e^{R^{F}_{t,T}}d\tilde{X}_{t}$, and (\ref{logaim}) can be rewritten as
\begin{align}
& \limsup_{T\rightarrow\infty}\frac{1}{T}\mathbb{E}\left[\ln
\left(\frac{\alpha}{1-\alpha}\int_{0}^{T}e^{R^{F}_{t,T}}d\tilde{X}_{t}\right)\right]\nonumber\\
 =& \limsup_{T\rightarrow\infty}\frac{1}{T}\mathbb{E}\left[\ln
\int_{0}^{T}e^{R^{F}_{t,T}}d\tilde{X}_{t}\right] = \lambda_1
+\limsup_{T\rightarrow\infty}\frac{1}{T}\mathbb{E}\left[\ln\int_{0}^{T}e^{-\lambda_1
T + R^{F}_{t,T}}d\tilde{X}_{t}\right].\label{ADDLAMDA}
\end{align}

Since $d\langle W^{X}, W^{F}\rangle_{t} = \rho dt$, $W^{X}_{t} = \rho W^{F}_{t} + \sqrt{1-\rho^{2}}W^{\perp}_t$, where $W^{\perp}$ is a Brownian Motion independent to $W^{F}$. Denote
$\mathbb{E}_{W^{\perp}_{T}}$ as the expectation conditional on $\left(W^{\perp}_{s}\right)_{0\leq s\leq T}$ (the whole trajectory of $W^{\perp}$ until $T$). By Lemma \ref{ADDWF} below,
\begin{align}
\textstyle
&\limsup_{T\rightarrow\infty}\frac{1}{T}\mathbb{E}\left[\ln\int_{0}^{T}e^{-\lambda_1
T+R^{F}_{t,T}}d\tilde{X}_{t}\right]\nonumber\\
\leq& \limsup_{T\rightarrow\infty}\frac{1}{T}\mathbb{E}\left[\ln\int_{0}^{T}e^{R^{F}_{t,T} - \int_{t}^{T}\left(\frac{1}{2}\left(\SF\right)^{2}ds
+\SF dW^{F}_{s}\right)}e^{-\frac{1-\alpha}{2}\left(\SX\right)^{2}t-(1-\alpha)\SX\overline{\rho W^{F}_{t}}}d\tilde{X}_{t}\right]\nonumber\\
=& \limsup_{T\rightarrow\infty}\frac{1}{T}\mathbb{E}\left[\mathbb{E}_{W^{\perp}_{T}}\left[\ln\int_{0}^{T}e^{R^{F}_{t,T} - \int_{t}^{T}\left(\frac{1}{2}\left(\SF\right)^{2}ds
+\SF dW^{F}_{s}\right)}e^{-\frac{1-\alpha}{2}\left(\SX\right)^{2}t-(1-\alpha)\SX\overline{\rho W^{F}_{t}}}d\tilde{X}_{t}\right]\right]\label{tower}\\
\leq& \limsup_{T\rightarrow\infty}\frac{1}{T}\mathbb{E}\left[\ln\mathbb{E}_{W^{\perp}_{T}}\left[\int_{0}^{T}e^{R^{F}_{t,T} - \int_{t}^{T}\left(\frac{1}{2}\left(\SF\right)^{2}ds+\SF dW^{F}_{s}\right)}e^{-\frac{1-\alpha}{2}\left(\SX\right)^{2}t-(1-\alpha)\SX\overline{\rho W^{F}_{t}}}d\tilde{X}_{t}\right]\right],\label{Jensen1}
\end{align}
\normalsize
where (\ref{tower}) follows from the tower property of conditional expectation, and (\ref{Jensen1}) from Jensen's inequality.

Next, Lemma \ref{SuperM} below implies that $M_t = e^{R^{F}_{0,t} - \int_{0}^{t}\left(\frac{1}{2}\left(\SF\right)^{2}ds + \SF dW^{F}_{s}\right)}$ is a super-martingale with respect to the filtration generated by $\left(W^{F}_{s}\right)_{0\leq s\leq t}$ and $\left(W^{\perp}_{s}\right)_{0\leq s\leq T}$ (the present of $W^{F}$ and $W^{\perp}$, plus the future of $W^{\perp}$). Thus,
\begin{equation}
\mathbb{E}_{W^{\perp}_{T},W^{F}_{t}}\left[e^{R^{F}_{t,T} - \int_{t}^{T}\left(\frac{1}{2}\left(\SF\right)^{2}ds
+\SF dW^{F}_{s}\right)}\right] \leq 1, \text{ } \forall \text{ }0\leq t \leq T. \label{BofCE}
\end{equation}

In addition, since $A_{t} = \int_{0}^{t}e^{-\frac{1-\alpha}{2}\left(\SX\right)^{2}t-(1-\alpha)\SX\bar W^{F}_{t}}d\tilde{X}_{t}$ is an increasing process, (\ref{BofCE}) and Lemma \ref{IntCE} imply that
\begin{align}
&\mathbb{E}_{W^{\perp}_{T}}\left[\int_{0}^{T}e^{R^{F}_{t,T} - \int_{t}^{T}\left(\frac{1}{2}\left(\SF\right)^{2}ds
+\SF dW^{F}_{s}\right)}e^{-\frac{1-\alpha}{2}\left(\SX\right)^{2}t-(1-\alpha)\SX\overline{\rho W^{F}_{t}}}d\tilde{X}_{t}\right]\nonumber\\
\leq& \mathbb{E}_{W^{\perp}_{T}}\left[\int_{0}^{T}e^{-\frac{1-\alpha}{2}\left(\SX\right)^{2}t-(1-\alpha)\SX\overline{\rho W^{F}_{t}}}d\tilde{X}_{t}\right].\label{ADS}
\end{align}

Then, from (\ref{ADDLAMDA}) and (\ref{ADS}), it follows that
\begin{equation*}
\limsup_{T\rightarrow\infty}\frac{1}{T}\mathbb{E}\left[\ln F_{T}\right]\leq \lambda_1
+\limsup_{T\rightarrow\infty}\frac{1}{T}\mathbb{E}\left[\ln\mathbb{E}_{W^{\perp}_{T}}\left[\int_{0}^{T}e^{-\frac{1-\alpha}{2}\left(\SX\right)^{2}t-(1-\alpha)\SX\overline{\rho W^{F}_{t}}}d\tilde{X}_{t}\right]\right].
\end{equation*}

Then, Lemma \ref{intparts} below proves that
\begin{equation*}
\limsup_{T\rightarrow\infty}\frac{1}{T}\mathbb{E}\left[\ln\mathbb{E}_{W^{\perp}_{T}}\left[\int_{0}^{T}e^{-\frac{1-\alpha}{2}\left(\SX\right)^{2}t-(1-\alpha)\SX\overline{\rho W^{F}_{t}}}d\tilde{X}_{t}\right]\right] \leq 0,
\end{equation*}
whence $\displaystyle\limsup_{T\rightarrow\infty}\frac{1}{T}\mathbb{E}\left[\ln F_{T}\right]\leq \lambda_1$.

\noindent \textit{The case of power utility}. Define $\tilde{F}_t = F_{t} + \bar X_{t}$. Then $\tilde{F}_{t} \geq F_{t}$ and $\tilde{F}_{t} \geq \bar X_{t}$ for every $t\geq 0$. Thus, the ESR of $F$ is less than or equal to the ESR of $\tilde{F}$. Lemma \ref{UBofF} below shows that this upper bound is also less than or equal to $\lambda_1$.
\end{proof}

\begin{lemma}\label{ADDWF}
For $\lambda_1$ defined in Lemma \ref{less},
\begin{align*}
&\limsup_{T\rightarrow\infty}\frac{1}{T}\mathbb{E}\left[\ln\int_{0}^{T}e^{-\lambda_1
T}e^{R^{F}_{t,T}}d\tilde{X}_{t}\right]\\
\leq& \limsup_{T\rightarrow\infty}\frac{1}{T}\mathbb{E}\left[\ln\int_{0}^{T}e^{R^{F}_{t,T} - \int_{t}^{T}\left(\frac{1}{2}\left(\SF\right)^{2}ds
+\SF dW^{F}_{s}\right)}e^{-\frac{1-\alpha}{2}\left(\SX\right)^{2}t-(1-\alpha)\SX \overline{\rho W^{F}_{t}}}d\tilde{X}_{t}\right].
\end{align*}
\end{lemma}
\begin{proof}
Define the stochastic process $N^{T}_{s} = W^{F}_{T} - W^{F}_{T-s}$, for $0\leq s\leq T$, and note that $N^{T}_{s}$ has the same distribution as $W^{F}_{s}$. It follows that
\begin{align*}
&\lim_{T \rightarrow \infty}\frac{1}{T}\mathbb{E}\left[-\SF \bar N^{T}_{T}-(1-\alpha)\SX\overline{\rho W^{F}_{T}}\right]\\
 =& \lim_{T \rightarrow \infty}\frac{1}{T}\mathbb{E}\left[-\SF \bar N^{T}_{T}\right] + \lim_{T \rightarrow \infty}\frac{1}{T}\mathbb{E}\left[-(1-\alpha)\SX\overline{\rho W^{F}_{T}}\right]\\
=&\lim_{T \rightarrow \infty}\frac{1}{T}\mathbb{E}\left[-\SF \bar W^{F}_{T}\right] + \lim_{T \rightarrow \infty}\frac{1}{T}\mathbb{E}\left[-(1-\alpha)\SX|\rho| \bar W^{F}_{T}\right]=0,
\end{align*}
where the last equality uses the fact that, for $a$, $b\neq 0$ and Brownian Motion $W$,
\begin{equation}
\lim_{T \rightarrow \infty}\frac{1}{b T}\mathbb{E}\left[a\bar W_{T}\right]
=\lim_{T \rightarrow \infty}\frac{1}{b T}\int_{0}^{\infty}\frac{a x}{\sqrt{2\pi T}}e^{-\frac{x^{2}}{2T}}dx = 0. \label{LN}
\end{equation}

Thus,
\begin{align}
&\limsup_{T\rightarrow\infty}\frac{1}{T}\mathbb{E}\left[\ln\int_{0}^{T}e^{-\lambda_1
T+ R^{F}_{t,T}}d\tilde{X}_{t}\right]\nonumber\\
=&\limsup_{T\rightarrow\infty}\frac{1}{T}\mathbb{E}\left[\ln\int_{0}^{T}e^{-\lambda_1
T+R^{F}_{t,T}}d\tilde{X}_{t}\right] +\lim_{T \rightarrow \infty}\frac{1}{T}\mathbb{E}\left[-\SF\bar N^{T}_{T}-(1-\alpha)\SX\overline{\rho W^{F}_{T}}\right]\nonumber\\
\leq&\limsup_{T\rightarrow\infty}\frac{1}{T}\mathbb{E}\left[\ln\int_{0}^{T}e^{-\lambda_1
T-\SF\bar N^{T}_{T}-(1-\alpha)\SX\overline{\rho W^{F}_{T}} + R^{F}_{t,T}}d\tilde{X}_{t}\right].\label{BADD}
\end{align}

Now, note that $\bar N^{T}_{T} \geq W^{F}_{T} - W^{F}_{t}$,
$\lambda_1 T \geq \frac{1}{2}\left(\SF\right)^{2}(T-t) + \frac{1-\alpha}{2}\left(\SX\right)^{2}t$ and
$(1-\alpha)\SX\overline{\rho W^{F}_{T}} \geq (1-\alpha)\SX\overline{\rho W^{F}_{t}}$, for all $0\leq t\leq T$. Thus,
\begin{align*}
&-\lambda_1
T-\SF \bar N^{T}_{T}-(1-\alpha)\SX\overline{\rho W^{F}_{T}} + R^{F}_{t,T}\\
\leq& -\frac{1}{2}\left(\SF\right)^{2}(T-t) - \frac{1-\alpha}{2}\left(\SX\right)^{2}t - \SF\left(W^{F}_{T} - W^{F}_{t}\right) -(1-\alpha)\SX\overline{\rho W^{F}_{t}} + R^{F}_{t,T}\\
=&R^{F}_{t,T} - \int_{t}^{T}\left(\frac{1}{2}\left(\SF\right)^{2}ds
+\SF dW^{F}_{s}\right)-\frac{1-\alpha}{2}\left(\SX\right)^{2}t-(1-\alpha)\SX\overline{\rho W^{F}_{t}}.
\end{align*}

Plugging this inequality into (\ref{BADD}) yields
\begin{align*}
&\limsup_{T\rightarrow\infty}\frac{1}{T}\mathbb{E}\left[\ln\int_{0}^{T}e^{-\lambda_1
T}e^{R^{F}_{t,T}}d\tilde{X}_{t}\right]\\
\leq& \limsup_{T\rightarrow\infty}\frac{1}{T}\mathbb{E}\left[\ln\int_{0}^{T}e^{R^{F}_{t,T} - \int_{t}^{T}\left(\frac{1}{2}\left(\SF\right)^{2}ds
+\SF dW^{F}_{s}\right)-\frac{1-\alpha}{2}\left(\SX\right)^{2}t-(1-\alpha)\SX\overline{\rho W^{F}_{t}}}d\tilde{X}_{t}\right].
\end{align*}
\end{proof}

\begin{lemma}\label{SuperM}
Let $\pi = \left(\pi_{t}\right)_{t\geq 0}$ be adapted to $\left\{\mathcal{F}_{t}\right\}_{t\geq0}$, the filtration generated by $\left(W^{X}_{s}\right)_{0\leq s\leq t}$ and $\left(W^{F}_{s}\right)_{0\leq s\leq t}$. Define
$\left\{\mathcal{G}_{t}\right\}_{t \geq 0}$ as the filtration generated by $\left(W^{F}_{s}\right)_{s\leq t}$ and $\left(W^{X}_{s}\right)_{s\leq T}$. Then
$M_{t} = e^{\int_{0}^{t}\left(\pi_{s} dW^{F}_{s} - \frac{1}{2}\pi_{s}^{2}ds\right)}$ is a super-martingale with respect to $\left\{\mathcal{G}_{t}\right\}_{t\geq0}$.
\end{lemma}

\begin{proof} Suppose $\pi$ is a simple process, i.e. $\pi_{t} = \displaystyle\Sigma^{n}_{i=1}\tilde{\pi}_{i}\textbf{1}_{\left[t_{i-1},t_{i}\right)}$ for a partition of [0,T], $0=t_{0}<t_{1}<t_{2}\cdots<t_{n}=T$ and $\tilde{\pi}_{i}$ is $\mathcal{F}_{t_{i-1}}$-measurable, for $i=1,\cdots,n$. Then for any $0\leq s < t \leq T$,
\begin{equation*}
\mathbb{E}_{W^{X}_{T},W^{F}_{s}}\left[e^{\int_{0}^{t}\left(\pi_{u} dW^{F}_{u} - \frac{1}{2}\pi_{u}^{2}du\right)}\right] = e^{\int_{0}^{s}\left(\pi_{u} dW^{F}_{u} - \frac{1}{2}\pi_{u}^{2}du\right)}\mathbb{E}_{W^{X}_{T},W^{F}_{s}}\left[e^{\int_{s}^{t}\left(\pi_{u} dW^{F}_{u} - \frac{1}{2}\pi_{u}^{2}du\right)}\right].
\end{equation*}
Since there exists $1 \leq k_{s} \leq k_{t} \leq n$ such that $t_{k_{s}-1} \leq s \leq t_{k_{s}}$ and $t_{k_{t}-1} \leq t \leq t_{k_{t}}$, $s$, $t$ and all the division points in between form a partition of $[s,t]$, denoted by $s =u_{0}<u_{1}<u_{2}\cdots<u_{m}=t$, with $\pi_t$ equal to a constant $\bar{\pi}_j$ (= $\tilde{\pi}_i$ for some $k_s - 1\leq i \leq k_t$) on $[u_{j-1},u_j)$, for $j = 1,\dots, m$. Then, since $W^F$ has independent increments,
\begin{align*}
&\mathbb{E}_{W^{X}_{T},W^{F}_{s}}\left[e^{\int_{s}^{t}\left(\pi_{u} dW^{F}_{u} - \frac{1}{2}\pi_{u}^{2}du\right)}\right] = \mathbb{E}_{W^{X}_{T},W^{F}_{s}}\left[e^{\sum\limits_{j=1}^{m}\bar{\pi}_{j}\left(W^{F}_{u_{j}}-W^{F}_{u_{j-1}}\right) - \frac{1}{2}\bar{\pi}_{j}^{2}\left(u_{j}-u_{j-1}\right)}\right]\\
=&\prod_{j=1}^{m}\mathbb{E}_{W^{X}_{T},W^{F}_{s}}\left[e^{\bar{\pi}_{j}\left(W^{F}_{u_{j}}-W^{F}_{u_{j-1}}\right) - \frac{1}{2}\bar{\pi}_{j}^{2}\left(u_{j}-u_{j-1}\right)}\right] \leq 1,
\end{align*}
and thus, $\mathbb{E}_{W^{X}_{T},W^{F}_{s}}\left[e^{\int_{0}^{t}\left(\pi_{u} dW^{F}_{u} - \frac{1}{2}\pi_{u}^{2}du\right)}\right]\leq e^{\int_{0}^{s}\left(\pi_{u} dW^{F}_{u} - \frac{1}{2}\pi_{u}^{2}du\right)}$.

For a general $\pi$, from the definition of stochastic integral, there exists a sequence of simple processes $\{\pi^{n}\}_{n=1}^{\infty}$, such that
\begin{equation*}
\int_{0}^{t}\left(\pi^{n}_{s} dW^{F}_{s} - \frac{1}{2}\left(\pi_{s}^{n}\right)^{2}ds\right)  \stackrel{n\uparrow\infty}{\longrightarrow}
\int_{0}^{t}\left(\pi_{s} dW^{F}_{s} - \frac{1}{2}\left(\pi_{s}\right)^{2}ds\right) \text{ a.s.}
\end{equation*}
Hence, for $0\leq s\leq t\leq T$,
\begin{align*}
&\mathbb{E}_{W^{X}_{T},W^{F}_{s}}\left[e^{\int_{0}^{t}\left(\pi_{u} dW^{F}_{u} - \frac{1}{2}\pi_{u}^{2}du\right)}\right] =\mathbb{E}_{W^{X}_{T},W^{F}_{s}}\left[\liminf_{n\rightarrow \infty}e^{\int_{0}^{t}\left(\pi^{n}_{u} dW^{F}_{u} - \frac{1}{2}\left(\pi^{n}_{u}\right)^{2}du\right)}\right].
\end{align*}
By Fatou's Lemma, this is less than or equal to
\begin{align*}
&\liminf_{n\rightarrow \infty}\mathbb{E}_{W^{X}_{T},W^{F}_{s}}\left[e^{\int_{0}^{t}\left(\pi^{n}_{u} dW^{F}_{u} - \frac{1}{2}\left(\pi_{u}^{n}\right)^{2}du\right)}\right]\\
\leq& \liminf_{n\rightarrow \infty} e^{\int_{0}^{s}\left(\pi^n_{u} dW^{F}_{u} - \frac{1}{2}\left(\pi^n_{u}\right)^{2}du\right)} = e^{\int_{0}^{s}\left(\pi_{u} dW^{F}_{u} - \frac{1}{2}\pi_{u}^{2}du\right)},
\end{align*}
which confirms that $M_{t}$ is a super-martingale with respect to $\left\{\mathcal{G}_{t}\right\}_{t\geq 0}$.
\end{proof}

\begin{lemma}\label{intparts}
$\displaystyle\limsup_{T\rightarrow\infty}\frac{1}{T}\mathbb{E}\left[\ln\mathbb{E}_{W^{\perp}_{T}}\left[\int_{0}^{T}e^{-\frac{1-\alpha}{2}\left(\SX\right)^{2}t-(1-\alpha)\SX\overline{\rho W^{F}_{t}}}d\tilde{X}_{t}\right]\right] \leq 0$.
\end{lemma}

\begin{proof}
By integration by parts,
\begin{align}
&\int_{0}^{T}e^{-\frac{1-\alpha}{2}\left(\SX\right)^{2}t-(1-\alpha)\SX\overline{\rho W^{F}_{t}}}d\tilde{X}_{t}=\frac{(1-\alpha)F_{0}}{\alpha} + \int_{0}^{T}e^{-\frac{1-\alpha}{2}\left(\SX\right)^{2}t-(1-\alpha)\SX\overline{\rho W^{F}_{t}}}d\bar X_{t}\nonumber\\
=&\frac{(1-\alpha)F_{0}}{\alpha} + e^{-\frac{1-\alpha}{2}\left(\SX\right)^{2}T-(1-\alpha)\SX\overline{\rho W^{F}_{T}}}\bar X_{T} - X_{0}\nonumber\\
+ &\frac{1-\alpha}{2}\left(\SX\right)^{2}\int_{0}^{T}e^{-\frac{1-\alpha}{2}\left(\SX\right)^{2}t-(1-\alpha)\SX\overline{\rho W^{F}_{t}}}\bar X_{t}dt\nonumber\\
+& (1-\alpha)\SX\int_{0}^{T}e^{-\frac{1-\alpha}{2}\left(\SX\right)^{2}t-(1-\alpha)\SX\overline{\rho W^{F}_{t}}}\bar X_{t}d\left(\overline{\rho W^{F}_t}\right).\label{intout}
\end{align}
Since $\bar X_{t} = X_0e^{(1-\alpha)\bar R^{X}_{t}}$, $e^{-\frac{1-\alpha}{2}\left(\SX\right)^{2}t-(1-\alpha)\SX\overline{\rho W^{F}_{t}}}\bar X_{t} \leq X_0e^{(1-\alpha)\left(\overline{R^{X}_{\cdot} - \frac{\left(\SX\right)^{2}}{2}\cdot - \SX \rho W^{F}_{\cdot}}\right)_{t}}$, following from Lemma \ref{maxmin}. Thus, from (\ref{intout}),
\begin{align}
&\int_{0}^{T}e^{-\frac{1-\alpha}{2}\left(\SX\right)^{2}t-(1-\alpha)\SX\bar W^{F}_{t}}d\tilde{X}_{t}\nonumber\\
\leq& \frac{(1-\alpha)F_{0}}{\alpha} + X_0e^{(1-\alpha)\left(\overline{R^{X}_{\cdot} - \frac{1}{2}\left(\SX\right)^{2}\cdot - \SX \rho W^{F}_{\cdot}}\right)_{T}}\nonumber\\
+&  \frac{1-\alpha}{2}\left(\SX\right)^{2}X_0\int_{0}^{T}e^{(1-\alpha)\left(\overline{R_{\cdot} - \frac{1}{2}\left(\SX\right)^{2}\cdot - \SX \rho W^{F}_{\cdot}}\right)_{t}}dt\nonumber\\
+& (1-\alpha)\SX X_0\int_{0}^{T}e^{(1-\alpha)\left(\overline{R^{X}_{\cdot} - \frac{1}{2}\left(\SX\right)^{2}\cdot - \SX \rho W^{F}_{\cdot}}\right)_{t}}d\left(\overline{\rho W^{F}_t}\right).\label{Bt2T1}
\end{align}
Since $e^{(1-\alpha)\left(\overline{R^{X}_{\cdot} - \frac{1}{2}\left(\SX\right)^{2}\cdot - \SX \rho W^{F}_{\cdot}}\right)_{t}} \leq e^{(1-\alpha)\left(\overline{R^{X}_{\cdot} - \frac{1}{2}\left(\SX\right)^{2}\cdot - \SX \rho W^{F}_{\cdot}}\right)_{T}}$ for every $t\leq T$,
\small
\begin{align}
&(1-\alpha)X_0\left(\frac{\left(\SX\right)^{2}}{2}\int_{0}^{T}e^{(1-\alpha)\left(\overline{R_{\cdot} - \frac{1}{2}\left(\SX\right)^{2}\cdot - \SX \rho W^{F}_{\cdot}}\right)_{t}}dt + \SX\int_{0}^{T}e^{(1-\alpha)\left(\overline{R^{X}_{\cdot} - \frac{1}{2}\left(\SX\right)^{2}\cdot - \SX \rho W^{F}_{\cdot}}\right)_{t}}d\left(\overline{\rho W^{F}_t}\right)\right)\nonumber\\
\leq &(1-\alpha)X_0\left(\frac{\left(\SX\right)^{2}}{2}e^{(1-\alpha)\left(\overline{R^{X}_{\cdot} - \frac{1}{2}\left(\SX\right)^{2}\cdot - \SX\rho  W^{F}_{\cdot}}\right)_{T}}T
+ \SX e^{(1-\alpha)\left(\overline{R^{X}_{\cdot} - \frac{1}{2}\left(\SX\right)^{2}\cdot - \SX \rho W^{F}_{\cdot}}\right)_{T}}\overline{\rho W^{F}_{T}}\right)\nonumber\\
=&\left(\frac{1-\alpha}{2}\left(\SX\right)^{2}T + (1-\alpha)\SX|\rho|\bar W^{F}_{T}\right)X_0e^{(1-\alpha)\left(\overline{R^{X}_{\cdot} - \frac{1}{2}\left(\SX\right)^{2}\cdot - \SX \rho W^{F}_{\cdot}}\right)_{T}}. \label{At2T1}
\end{align}
\normalsize
Thus, from (\ref{Bt2T1}) and (\ref{At2T1}),
\small
\begin{align}
&\mathbb{E}_{W^{\perp}_{T}}\left[\int_{0}^{T}e^{-\frac{1-\alpha}{2}\left(\SX\right)^{2}t-(1-\alpha)\SX\overline{\rho W^{F}_{t}}}d\tilde{X}_{t}\right]\nonumber\\
\leq&\mathbb{E}_{W^{\perp}_{T}}\left[\frac{(1-\alpha)F_{0}}{\alpha } + \left(1 + \frac{1-\alpha}{2}\left(\SX\right)^{2}T + (1-\alpha)\SX|\rho| \bar W^{F}_{T}\right)X_0e^{(1-\alpha)\left(\overline{R^{X}_{\cdot} - \frac{1}{2}\left(\SX\right)^{2}\cdot - (1-\alpha)\SX \rho W^{F}_{\cdot}}\right)_{T}}\right]\nonumber\\
=&\frac{(1-\alpha)F_{0}}{\alpha} + \mathbb{E}_{W^{\perp}_{T}}\left[\left(1 + \frac{1-\alpha}{2}\left(\SX\right)^{2}T + (1-\alpha)\SX|\rho| \bar W^{F}_{T}\right)X_0e^{(1-\alpha)\left(\overline{R^{X}_{\cdot} - \frac{1}{2}\left(\SX\right)^{2}\cdot - (1-\alpha)\SX \rho W^{F}_{\cdot}}\right)_{T}}\right]\nonumber\\
\leq&\frac{(1-\alpha)F_{0}}{\alpha} + \mathbb{E}_{W^{\perp}_{T}}\left[L^{\delta}_{T}\right]^{\frac{1}{\delta}}\mathbb{E}_{W^{\perp}_{T}}\left[K_{T}^{\frac{\delta}{\delta -1}}\right]^{\frac{\delta-1}{\delta}},\label{int11}
\end{align}
\normalsize
for any $\delta > 1$, by H\"{o}lder's inequality, where
\begin{align*}
K_{T} =& \left(1 + \frac{1-\alpha}{2}\left(\SX\right)^{2}T + (1-\alpha)\SX|\rho|\bar W^{F}_{T}\right)X_0,\\
L_{T} =& e^{(1-\alpha)\left(\overline{R^{X}_{\cdot} - \frac{1}{2}\left(\SX\right)^{2}\cdot - (1-\alpha)\rho\SX  W^{F}_{\cdot}}\right)_{T}}
\end{align*}

Since $\delta >1$ and $\frac{\delta}{\delta-1}>1$,
\begin{align}
&\mathbb{E}_{W^{\perp}_{T}}\left[K_{T}^{\frac{\delta}{\delta -1}}\right]^{\frac{\delta-1}{\delta}}
\leq \left(1 + \frac{1-\alpha}{2}\left(\SX\right)^{2}T + (1-\alpha)\SX|\rho|\mathbb{E}_{W^{\perp}_{T}}\left[\left(\bar W^{F}_{T}\right)^{\frac{\delta}{\delta -1}}\right]^{\frac{\delta-1}{\delta}}\right)X_0\nonumber\\
=& \left(1 + \frac{1-\alpha}{2}\left(\SX\right)^{2}T + \sqrt{2}(1-\alpha)\SX|\rho|\left(\frac{\Gamma(\frac{1+\frac{\delta}{\delta -1}}{2})}{\sqrt{\pi}}\right)^{\frac{\delta-1}{\delta}}\sqrt{T}\right)X_0,\label{int12}
\end{align}
following from Minkowski inequality
$\left(\mathbb{E}\left[(f + g)^{p}\right]^{\frac{1}{p}} \leq \mathbb{E}\left[f^{p}\right]^{\frac{1}{p}} + \mathbb{E}\left[f^{p}\right]^{\frac{1}{p}}\right)$.

Then (\ref{int11}) and (\ref{int12}) imply that for any $\delta >1$,
\begin{equation*}
\mathbb{E}_{W^{\perp}_{T}}\left[\int_{0}^{T}e^{-\frac{1-\alpha}{2}\left(\SX\right)^{2}t-(1-\alpha)\SX\overline{\rho W^{F}_{t}}}d\tilde{X}_{t}\right]\leq \frac{(1-\alpha)F_{0}}{\alpha} + C_{T}\mathbb{E}_{W^{\perp}_{T}}\left[L_T^{\delta}\right]^{\frac{1}{\delta}},
\end{equation*}
where $C_{T} = \left(1 + \frac{1-\alpha}{2}\left(\SX\right)^{2}T + \sqrt{2}(1-\alpha)\SX|\rho|\left(\frac{\Gamma(\frac{1+\frac{\delta}{\delta -1}}{2})}{\sqrt{\pi}}\right)^{\frac{\delta-1}{\delta}}\sqrt{T}\right)X_0$.

Thus
\begin{align}
&\limsup_{T\rightarrow\infty}\frac{1}{T}\mathbb{E}\left[\ln\mathbb{E}_{W^{\perp}_{T}}\left[\int_{0}^{T}e^{-\frac{1-\alpha}{2}\left(\SX\right)^{2}t-(1-\alpha)\SX\overline{\rho W^{F}_{t}}}d\tilde{X}_{t}\right]\right]\nonumber\\
\leq& \limsup_{T\rightarrow\infty}\frac{1}{T}\mathbb{E}\left[\ln\left(\frac{(1-\alpha)F_{0}}{\alpha} + C_{T}\mathbb{E}_{W^{\perp}_{T}}\left[L^{\delta}_T\right]^{\frac{1}{\delta}}\right)\right]\nonumber\\
\leq&\limsup_{T\rightarrow\infty}\frac{1}{T}\ln C_{T} + \limsup_{T\rightarrow\infty}\frac{1}{T}\mathbb{E}\left[\ln\left(\frac{(1-\alpha)F_{0}}{\alpha  C_T} + \mathbb{E}_{W^{\perp}_{T}}\left[L^{\delta}_T\right]^{\frac{1}{\delta}}\right)\right]\label{parts1}.
\end{align}
Note that the limit in the first term is 0. Furthermore, since $\frac{(1-\alpha)F_{0}}{\alpha C_T}\rightarrow 0$ as $T \uparrow \infty$ and $L_T\geq 1$, for $T$ large enough, $\frac{(1-\alpha)F_{0}}{\alpha C_T} < 1\leq  \mathbb{E}_{W^{\perp}_{T}}\left[L^{\delta}_T\right]^{\frac{1}{\delta}}$. Thus, (\ref{parts1}) is less than or equal to
\begin{equation}
\limsup_{T\rightarrow\infty}\frac{1}{T}\mathbb{E}\left[\ln\left(2 \mathbb{E}_{W^{\perp}_{T}}\left[L^{\delta}_T\right]^{\frac{1}{\delta}}\right)\right]
=\limsup_{T\rightarrow\infty}\frac{1}{\delta T}\mathbb{E}\left[\ln\mathbb{E}_{W^{\perp}_{T}}\left[ L^{\delta}_T\right]\right].\label{parts2}
\end{equation}
Since (\ref{LN}) implies that $\displaystyle\liminf_{T\rightarrow\infty}\frac{1}{\delta T}\mathbb{E}\left[-(1-\alpha)\sqrt{1-\rho^{2}}\delta\SX\bar W^{\perp}_{T}\right] = 0$, (\ref{parts2}) equals
\begin{align}
&\limsup_{T\rightarrow\infty}\frac{1}{\delta T}\mathbb{E}\left[\ln\mathbb{E}_{W^{\perp}_{T}}\left[ e^{(1-\alpha)\delta\left(\overline{R^{X}_{\cdot} - \frac{1}{2}\left(\SX\right)^{2}\cdot - \rho\SX  W^{F}_{\cdot}}\right)_{T}}\right]\right]\nonumber\\
+&\liminf_{T\rightarrow\infty}\frac{1}{\delta T}\mathbb{E}\left[-(1-\alpha)\sqrt{1-\rho^{2}}\delta\SX\bar W^{\perp}_{T}\right]\nonumber\\
\leq&\limsup_{T\rightarrow\infty}\frac{1}{\delta T}\mathbb{E}\left[\ln\mathbb{E}_{W^{\perp}_{T}}\left[ e^{(1-\alpha)\delta\left(\overline{R^{X}_{\cdot} - \frac{1}{2}\left(\SX\right)^{2}\cdot - \rho\SX W^{F}_{\cdot}}\right)_{T}-(1-\alpha)\sqrt{1-\rho^{2}}\delta\SX\bar W^{\perp}_{T}}\right]\right].\label{parts3}
\end{align}

Then, by Lemma \ref{maxmin}, the running maximum and running minimum can be combined, and (\ref{parts3}) is less than or equal to
\begin{align}
&\limsup_{T\rightarrow\infty}\frac{1}{\delta T}\mathbb{E}\left[\ln\mathbb{E}_{W^{\perp}_{T}}\left[ e^{(1-\alpha)\delta\left(\overline{R^{X}_{\cdot} - \frac{1}{2}\left(\SX\right)^{2}\cdot - \rho\SX  W^{F}_{\cdot}-\sqrt{1-\rho^{2}}\SX W^{\perp}_{\cdot}}\right)_{T}}\right]\right]\nonumber\\
=&\limsup_{T\rightarrow\infty}\frac{1}{\delta T}\mathbb{E}\left[\ln\mathbb{E}_{W^{\perp}_{T}}\left[ e^{(1-\alpha)\delta\left(\overline{R^{X}_{\cdot} - \frac{1}{2}\left(\SX\right)^{2}\cdot - \SX W^{X}_{\cdot}}\right)_{T}}\right]\right]\nonumber\\
\leq&\limsup_{T\rightarrow\infty}\frac{1}{\delta T}\ln\mathbb{E}\left[e^{(1-\alpha)\delta\left(\overline{R^{X}_{\cdot} - \frac{1}{2}\left(\SX\right)^{2}\cdot - \SX W^{X}_{\cdot}}\right)_{T}}\right],\label{Jensen2}
\end{align}
where (\ref{Jensen2}) follows from Jensen's inequality and the tower property of conditional expectation.

$M_{t} = e^{R^{X}_{t} - \frac{1}{2}\left(\SX\right)^{2}t - \SX W^{X}_{t}}$ is a local martingale with respect to the filtration generated by $\left(W^{F}_{s}\right)_{0\leq s \leq t}$ and $\left(W^{\perp}_{s}\right)_{0\leq s \leq t}$. Then, since $\bar M_{t} \leq \bar M_{\infty}$, which in turn is dominated by a random variable $X$, and $X^{-1}$ is uniformly distributed on $[0,1]$ (cf. $(54)$ in \citet*{guasoni2012nopart}), for $1<\delta < \frac{1}{1-\alpha}$, (\ref{Jensen2}) is less than or equal to
\begin{align*}
&\limsup_{T\rightarrow\infty}\frac{1}{\delta T}\ln\mathbb{E}\left[\left(\bar M_{\infty}\right)^{(1-\alpha)\delta}\right]\leq\limsup_{T\rightarrow\infty}\frac{1}{\delta T}\ln\left(\int_{0}^{1}x^{-(1-\alpha)\delta}dx\right)\\
=&\limsup_{T\rightarrow\infty}\frac{1}{\delta T}\ln \left(\frac{1}{1-(1-\alpha)\delta}\right)=0,
\end{align*}
which concludes the proof.
\end{proof}

For the rest of this paper, let $p = 1-\gamma$.

\begin{lemma}\label{UBofF}
$\displaystyle\limsup_{T\rightarrow\infty}\frac{1}{T}\ln\mathbb{E}\left[\tilde{F}^{p}_{T}\right]^{\frac{1}{p}}\leq \lambda_1$ for any $0<p<1$.
\end{lemma}

\begin{proof}
Let $\tilde{\pi}^{F}_{t} = \frac{\left(\tilde{F}_{t} - \bar X_{t}\right)}{\tilde{F}_{t}}\pif_{t}$. Investing
$\pif_{t}$ of $\tilde{F}_{t} - \bar X_{t}$ in the risky asset is equivalent to investing $\tilde{\pi}^{F}_{t}$ of $\tilde{F}_{t}$ and thus $\tilde{\pi}^{F}_{t}$ can be regarded as an investment strategy for $\tilde{F}_{t}$. Then
\begin{align*}
d\tilde{F}^{p}_{t} &= p\tilde{F}_{t}^{p-1}\left(\tilde{F}_{t} - \bar X_{t}\right)\left(\pif_{t}\muf dt +\pif_{t}\sigmaf dW^{F}\right)\\
 &+ \frac{p(p-1)}{2}\tilde{F}^{p-2}_{t} \left(\tilde{F}_{t} - \bar X_{t}\right)^{2}\left(\pif_{t}\sigmaf\right)^{2}dt + \frac{1}{1-\alpha}p\tilde{F}_{t}^{p-1}d\bar X_{t}\\
&= p\tilde{F}_{t}^{p}\left(\left(\tilde{\pi}^{F}_{t}\muf + \frac{(p-1)}{2}\left(\tilde{\pi}^{F}_{t}\sigmaf\right)^{2}\right)dt + \tilde{\pi}^{F}_{t}\sigmaf dW^{F}\right) + \frac{1}{1-\alpha}p\tilde{F}_{t}^{p-1}d\bar X_{t}.
\end{align*}

Solving this differential equation, $\tilde{F}^{p}_{T} = F_{0}^{p}e^{pR^{F,\tilde{\pi}^{F}}_{T}} + \frac{1}{1-\alpha}\int_{0}^{T}pe^{pR^{F,\tilde{\pi}^{F}}_{t,T}}\tilde{F}^{p-1}_{t}d\bar X_{t}$. Thus,
\begin{equation*}
\limsup_{T\rightarrow\infty}\frac{1}{T}\ln\mathbb{E}\left[\tilde{F}^{p}_{T}\right]^{\frac{1}{p}}=\limsup_{T\rightarrow\infty}\frac{1}{T}\ln\mathbb{E}\left[F^{p}_{0}e^{pR^{F,\tilde{\pi}^{F}}_{T}} + \frac{1}{1-\alpha}\int_{0}^{T}pe^{pR^{F,\tilde{\pi}^{F}}_{t,T}}\tilde{F}^{p-1}_{t}d\bar X_{t}\right]^{\frac{1}{p}}.
\end{equation*}
Since $0<p<1$, from \citet*{Dembolargedeviation}, Lemma $1.2.15$, for any positive processes $\left(f_{t}\right)_{t\geq 0}$ and $\left(g_{t}\right)_{t\geq 0}$,
\begin{equation*}
\limsup_{T\rightarrow\infty}\frac{1}{T}\ln\left(f_{T} + g_{T}\right)^{\frac{1}{p}} = \max\left(\limsup_{T\rightarrow\infty}\frac{1}{T}\ln f_{T}^{\frac{1}{p}}, \limsup_{T\rightarrow\infty}\frac{1}{T}\ln g_{T}^{\frac{1}{p}}\right).
\end{equation*}
It follows that,
\small
\begin{align}
&\limsup_{T\rightarrow\infty}\frac{1}{T}\ln\mathbb{E}\left[F^{p}_{0}e^{pR^{F,\tilde{\pi}^{F}}_{T}} + \frac{1}{1-\alpha}\int_{0}^{T}pe^{pR^{F,\tilde{\pi}^{F}}_{t,T}}\tilde{F}^{p-1}_{t}d\bar X_{t}\right]^{\frac{1}{p}}\nonumber\\
=& \max\left(\limsup_{T\rightarrow\infty}\frac{1}{T}\ln \mathbb{E}\left[F^{p}_{0}e^{pR^{F,\tilde{\pi}^{F}}_{T}}\right]^{\frac{1}{p}}, \limsup_{T\rightarrow\infty}\frac{1}{T}\ln \mathbb{E}\left[\frac{1}{1-\alpha}\int_{0}^{T}pe^{pR^{F,\tilde{\pi}^{F}}_{t,T}}\tilde{F}^{p-1}_{t}d\bar X_{t}\right]^{\frac{1}{p}}\right).\label{max}
\end{align}
\normalsize
Note that, since $0<p<1$, by H\"{o}lder's inequality,
\begin{equation*}
\mathbb{E}\left[F^{p}_{0}e^{pR^{F,\tilde{\pi}^{F}}_{T}}\right]^{\frac{1}{p}}\mathbb{E}\left[e^{q\left(-\SF W^{F}_{T} - \frac{\left(\SF\right)^{2}}{2}T\right)}\right]^{\frac{1}{q}} \leq \mathbb{E}\left[F_{0}e^{R^{F,\hat{\pi}^{F}}_{T}-\SF W^{F}_{T} - \frac{\left(\SF\right)^{2}}{2}T}\right]\leq F_{0},
\end{equation*}
where $q = \frac{p}{p-1}$, and it indicates that
\begin{equation}
\mathbb{E}\left[F_0^pe^{pR^{F,\tilde{\pi}^{F}}_{T}}\right]^{\frac{1}{p}} \leq F_0\mathbb{E}\left[e^{q\left(-\SF W^{F}_{T} - \frac{\left(\SF\right)^{2}}{2}T\right)}\right]^{-\frac{1}{q}} = F_0e^{\frac{\left(\SF\right)^{2}}{2(1-p)}T}.\label{WEXP}
\end{equation}
Thus,
\begin{equation}
\limsup_{T\rightarrow\infty}\frac{1}{T}\ln\mathbb{E}\left[F^{p}_{0}e^{pR^{F,\tilde{\pi}^{F}}_{T}}\right]^{\frac{1}{p}}
\leq \limsup_{T\rightarrow\infty}\frac{1}{T}\ln F_0e^{\frac{\left(\SF\right)^{2}}{2(1-p)}T}= \frac{\left(\SF\right)^{2}}{2(1-p)}.\label{max1}
\end{equation}
For the second term in (\ref{max}), since $p<1$, and $\tilde{F}_{t} \geq \bar X_{t}$,
$\tilde{F}^{p-1}_{t} \leq \bar X_{t}^{p-1}$, and
\begin{align*}
&\mathbb{E}\left[\int_{0}^{T}pe^{pR^{F,\tilde{\pi}^{F}}_{t,T}}e^{-p\lambda_1 T}\tilde{F}^{p-1}_{t}d\bar X_{t}\right]\leq\mathbb{E}\left[\int_{0}^{T}e^{pR^{F,\tilde{\pi}^{F}}_{t,T}-p\lambda_1(T-t)}e^{-p\lambda_1 t}p\bar X_{t}^{p-1}d\bar X_{t}\right]\\
=&\mathbb{E}\left[\int_{0}^{T}e^{pR^{F,\tilde{\pi}^{F}}_{t,T}-p\lambda_1(T-t)}e^{-p\lambda_1 t}d\bar X_{t}^{p}\right].
\end{align*}
(\ref{WEXP}) implies that $\mathbb{E}_{t}\left[e^{pR^{F,\tilde{\pi}^{F}}_{t,T}-p\lambda_1(T-t)}\right]\leq 1$. Then, since $\int_{0}^{T}e^{-p\lambda_1 t}d\bar X_{t}^{p}$ is an increasing process, Lemma \ref{IntCE} implies that 
\begin{equation*}
\mathbb{E}\left[\int_{0}^{T}e^{pR^{F,\tilde{\pi}^{F}}_{t,T}-p\lambda_1(T-t)}e^{-p\lambda_1 t}d\bar X_{t}^{p}\right]\leq  \mathbb{E}\left[\int_{0}^{T}e^{-p\lambda_1 t}d\bar X_{t}^{p}\right], 
\end{equation*}
 and hence
\begin{align}
&\limsup_{T\rightarrow\infty}\frac{1}{T}\ln\mathbb{E}\left[\frac{1}{1-\alpha}\int_{0}^{T}pe^{pR^{F,\tilde{\pi}}_{t,T}}\tilde{F}^{p-1}_{t}d\bar X_{t}\right]^{\frac{1}{p}}\nonumber\\
=& \limsup_{T\rightarrow\infty}\frac{1}{T}\ln\mathbb{E}\left[\int_{0}^{T}pe^{pR^{F,\tilde{\pi}}_{t,T}}\tilde{F}^{p-1}_{t}d\bar X_{t}\right]^{\frac{1}{p}}\nonumber\\
=&\lambda_1 + \limsup_{T\rightarrow\infty}\frac{1}{T}\ln\mathbb{E}\left[\int_{0}^{T}pe^{pR^{F,\tilde{\pi}}_{t,T}}e^{-p\lambda_1T}\tilde{F}^{p-1}_{t}d\bar X_{t}\right]^{\frac{1}{p}}\nonumber\\
\leq&\lambda_1 + \limsup_{T\rightarrow\infty}\frac{1}{T}\ln\mathbb{E}\left[\int_{0}^{T}e^{-p\lambda_1 t}d\bar X_{t}^{p}\right]^{\frac{1}{p}}.\label{Bintp}
\end{align}
Now, integration by parts implies that
\begin{align}
\int_{0}^{T}e^{-p\lambda_1 t}d\bar X_{t}^{p} &= e^{-p\lambda_1 T}\bar X_{T}^{p} - X_0^p + p\lambda_1\int_{0}^{T}e^{-p\lambda_1 t}\bar X_{t}^{p}dt\nonumber\\
&\leq e^{-p\lambda_1 T}\bar X_T^{p}+ p\lambda_1\int_{0}^{T}e^{-p\lambda_1 t}\bar X_t^{p}dt.\label{BMinMax}
\end{align}
Lemma \ref{maxmin} implies that $e^{-p\lambda_1 t}\bar X_t^{p} \leq X^p_0e^{(1-\alpha)p\left(\overline{R^{X}_{\cdot} - \frac{\lambda_1}{1-\alpha}\cdot}\right)_{t}}$ for every $0\leq t\leq T$. Thus (\ref{BMinMax}) is less than or equal to
\begin{align*}
&X_0^pe^{(1-\alpha)p\left(\overline{R^{X}_{\cdot} - \frac{\lambda_1}{1-\alpha}\cdot}\right)_{T}}+ X^p_0 p\lambda_1\int_{0}^{T}e^{(1-\alpha)p\left(\overline{R^{X}_{\cdot} - \frac{\lambda_1}{1-\alpha}\cdot}\right)_{t}}dt\\
\leq& X_0^pe^{(1-\alpha)p\left(\overline{R^{X}_{\cdot} - \frac{\lambda_1}{1-\alpha}\cdot}\right)_{T}} + X_0^pp\lambda_1 Te^{(1-\alpha)p\left(\overline{R^{X}_{\cdot} - \frac{\lambda_1}{1-\alpha}\cdot}\right)_{T}}\\
=&X_0^p\left(1+ p\lambda_1 T\right)e^{(1-\alpha)p\left(\overline{R^{X}_{\cdot} - \frac{\lambda_1}{1-\alpha}\cdot}\right)_{T}}.
\end{align*}
Now, Lemma 9 in \citet*{guasoni2012nopart} with $\varphi - r = \frac{\lambda_1}{1-\alpha}$ implies that
\begin{align}
&\limsup_{T\rightarrow\infty}\frac{1}{T}\ln\mathbb{E}\left[X_0^p(1+p\lambda_1 T)e^{p\left(\overline{R^{X}_{\cdot}-\frac{\lambda_1}{1-\alpha}\cdot}\right)_{T}}\right]^{\frac{1}{p}}\nonumber\\
=&\limsup_{T\rightarrow\infty}\frac{1}{T}\ln\mathbb{E}\left[e^{p\left(\overline{R^{X}_{\cdot}-\frac{\lambda_1}{1-\alpha}\cdot}\right)_{T}}\right]^{\frac{1}{p}}\leq 0.\label{GO}
\end{align}
Thus, (\ref{Bintp}) and (\ref{GO}) imply that
\begin{align}
&\limsup_{T\rightarrow\infty}\frac{1}{T}\ln\mathbb{E}\left[\frac{1}{1-\alpha}\int_{0}^{T}e^{pR^{F,\tilde{\pi}}_{t,T}}\tilde{F}^{p-1}_{t}d\bar X_{t}\right]^{\frac{1}{p}}\nonumber\\
\leq& \lambda_1 + \limsup_{T\rightarrow\infty}\frac{1}{T}\ln\mathbb{E}\left[\int_{0}^{T}e^{-p\lambda_1 t}d\bar X_{t}^{p}\right]^{\frac{1}{p}}\nonumber\\
\leq& \lambda_1 + \limsup_{T\rightarrow\infty}\frac{1}{T}\ln\mathbb{E}\left[(1+p\lambda_1 T)e^{p\left(\overline{R^{X}_{\cdot}-\frac{\lambda_1}{1-\alpha}\cdot}\right)_{T}}\right]^{\frac{1}{p}} \leq\lambda_1.\label{max2}
\end{align}
Then, (\ref{max}), (\ref{max1}), and (\ref{max2}) imply that
\begin{equation*}
\limsup_{T\rightarrow\infty}\frac{1}{T}\ln\mathbb{E}\left[\tilde{F}^{p}_{T}\right]^{\frac{1}{p}} \leq \max\left(\frac{\left(\SF\right)^{2}}{2(1-p)}, \lambda_1\right) = \lambda_1.
\end{equation*}
\end{proof}
To prove Theorem \ref{Main}, it now remains to show that the upper bound in Lemma \ref{less} is achieved by the ESR induced by the strategies in (\ref{optimal1}) and
(\ref{optimal2}), and hence they are optimal. Plugging them into the dynamics of $X_t$ and $F_t$, the corresponding fund's value and wealth processes follow
\begin{align*}
d\hat{X}_{t} &= \hat{X}_{t}\left(\frac{\left(\SX\right)^{2}}{1-(1-\alpha)p} dt +
\frac{\SX}{1-(1-\alpha)p} dW^{X}_{t}\right) - \frac{\alpha}{1-\alpha}d\bar{\hat{X}}_{t},\\
d\hat{F}_{t} &=
\left(\hat{F}_{t}-\frac{\alpha}{1-\alpha}(\bar{\hat{X}}_{t}-X_{0})\right)\left(\frac{\left(\SF\right)^{2}}{1-p}
dt + \frac{\SF}{1-p} dW^{F}_{t}\right) +\frac{\alpha}{1-\alpha}d\bar{\hat{X}}_{t}.
\end{align*}
Letting $\hat{R}^{X}_{T} = \frac{\left(1-2(1-\alpha)p\right)\left(\SX\right)^{2}}{2\left(1-(1-\alpha)p\right)^{2}}T +
\frac{\SX }{1-(1-\alpha)p}W^{X}_{T}$, $\hat{R}^{F}_{T} = \frac{(1-2p)\left(\SF\right)^{2}}{2(1-p)^{2}}T +
\frac{\SF}{1-(1-p} W^{F}_{T}$, and solving the above differential equations imply that
\begin{align*}
\bar{\hat{X}}_{T} &= X_{0}e^{(1-\alpha)\bar{\hat{R}}^{X}_{T}},\\
\hat{F}_{T} &= F_{0}e^{\hat{R}^{F}_{T}}+
\frac{\alpha}{1-\alpha}(\bar{\hat{X}}_{T}-X_{0}).
\end{align*}

\begin{lemma}\label{optimum}
For $\optpix$ and $\optpif$ in (\ref{optimal1}) and (\ref{optimal2}), $\ESR_{\gamma}\left(\hat{\pi}^{X},\hat{\pi}^{F}\right) = \lambda_1$, for any $0 < \gamma \leq 1$.
\end{lemma}

\begin{proof}
Let $G_{t} = F_{0}e^{\hat{R}^{F}_{T}}$ and $H_{t} =
\frac{\alpha}{1-\alpha}\left(\bar{\hat{X}}_{t}-X_{0}\right)$, then $\hat{F}_{T} = G_{T} + H_{T}$. From Lemma \ref{less}, it suffices to prove that $\ESR_{\gamma}\left(\hat{\pi}^{X},\hat{\pi}^{F}\right) \geq \lambda_1$.\\

\noindent \textit{The case of logarithmic utility} ($p=0$). Since $H_{t}$ is a positive process,
\begin{equation*}
\lim_{T\rightarrow\infty}\frac{1}{T}\mathbb{E}\left[\ln\hat{F}_{T}\right]=\lim_{T\rightarrow\infty}\frac{1}{T}\mathbb{E}\left[\ln(G_{T} + H_{T})\right]\geq\lim_{T\rightarrow\infty}\frac{1}{T}\mathbb{E}\left[\ln G_{T}\right]=\frac{\left(\SF\right)^{2}}{2}.
\end{equation*}

Likewise, since $G_{t}$ is a positive process, Lemma \ref{AC} below implies that
\begin{align*}
&\lim_{T\rightarrow\infty}\frac{1}{T}\mathbb{E}\left[\ln\hat{F}_{T}\right]\geq\lim_{T\rightarrow\infty}\frac{1}{T}\mathbb{E}\left[\ln H_{T}\right]=\lim_{T\rightarrow\infty}\frac{1}{T}\mathbb{E}\left[\ln \left(\bar{\hat{X}}_{T}-X_{0}\right)\right]\\
=& \lim_{T\rightarrow\infty}\frac{1}{T}\mathbb{E}\left[\ln \left(e^{(1-\alpha)\bar{\hat{R}}^{X}_{T}}-1\right)\right] = \lim_{T\rightarrow\infty}\frac{1}{T}\mathbb{E}\left[(1-\alpha)\bar{\hat{R}}^{X}_{T}\right]\\
=& \lim_{T\rightarrow\infty}\frac{1}{T}\mathbb{E}\left[(1-\alpha)\left(\overline{\frac{\left(\SX\right)^{2}}{2}\cdot +
\SX W^{X}_{\cdot}}\right)_{T}\right].
\end{align*}

Since $\left(\overline{\frac{\left(\SX\right)^{2}}{2}\cdot +
\SX W^{X}_{\cdot}}\right)_{T} \geq \frac{\left(\SX\right)^{2}T}{2} + \SX \underbar W^{X}_{T} = \frac{\left(\SX\right)^{2}T}{2} - \SX \overline{\left(-W^X\right)}_T$, which follows from Lemma \ref{maxmin},
\begin{align*}
&\lim_{T\rightarrow\infty}\frac{1}{T}\mathbb{E}\left[(1-\alpha)\left(\overline{\frac{\left(\SX\right)^{2}}{2}\cdot +
\SX W^{X}_{\cdot}}\right)_{T}\right]\\
\geq& \frac{(1-\alpha)\left(\SX\right)^{2} }{2} - \lim_{T\rightarrow\infty}\frac{1}{T}\mathbb{E}\left[(1-\alpha)
\SX \overline{-W}^{X}_{T}\right] = \frac{(1-\alpha)\left(\SX\right)^{2}}{2},
\end{align*}
where the last equality follows from (\ref{LN}).

Thus, $\displaystyle\lim_{T\rightarrow\infty}\frac{1}{T}\mathbb{E}\left[\ln\hat{F}_{T}\right]\geq \max\left(\frac{\left(\SF\right)^{2}}{2},\frac{(1-\alpha)\left(\SX\right)^{2}}{2}\right)=\lambda_1.$\\

\noindent \textit{The case of power utility} ($0<p<1$).
Since $H_{t}$ is a positive process,
\begin{equation*}
\lim_{T\rightarrow\infty}\frac{1}{T}\ln\mathbb{E}\left[\hat{F}^{p}_{T}\right]^{\frac{1}{p}}=\lim_{T\rightarrow\infty}\frac{1}{T}\ln\mathbb{E}\left[(G_{T} + H_{T})^{p}\right]^{\frac{1}{p}}\geq\lim_{T\rightarrow\infty}\frac{1}{T}\ln\mathbb{E}\left[G^{p}_{T}\right]^{\frac{1}{p}}=\frac{\left(\SF\right)^{2}}{2(1-p)}.
\end{equation*}

Likewise, since $G_{t}$ is a positive process, Lemma \ref{AC} below implies that
\begin{align*}
&\lim_{T\rightarrow\infty}\frac{1}{T}\ln\mathbb{E}\left[\hat{F}^{p}_{T}\right]^{\frac{1}{p}}\geq\lim_{T\rightarrow\infty}\frac{1}{T}\ln\mathbb{E}\left[H^{p}_{T}\right]^{\frac{1}{p}}\\
=&\lim_{T\rightarrow\infty}\frac{1}{T}\ln\mathbb{E}\left[\left(\bar{\hat{X}}_{T}-X_{0}\right)^{p}\right]^{\frac{1}{p}}=\lim_{T\rightarrow\infty}\frac{1}{T}\ln\mathbb{E}\left[\left(e^{(1-\alpha)\bar{\hat{R}}^{X}_{T}}-1\right)^{p}\right]^{\frac{1}{p}}\\
=&\lim_{T\rightarrow\infty}\frac{1}{T}\ln\mathbb{E}\left[e^{(1-\alpha)p\bar{\hat{R}}^{X}_{T}}\right]^{\frac{1}{p}}=\frac{(1-\alpha)\left(\SX\right)^{2}}{2\left(1-(1-\alpha)p\right)}.
\end{align*}
where the last equality follows from Lemma $11$ in \citet*{guasoni2012nopart}. Thus, $\displaystyle\lim_{T\rightarrow\infty}\frac{1}{T}\ln\mathbb{E}\left[\hat{F}^{p}_{T}\right]^{\frac{1}{p}}\geq \max\left(\frac{\left(\SF\right)^{2}}{2(1-p)},\frac{(1-\alpha)\left(\SX\right)^{2}}{2\left(1-(1-\alpha)p\right)}\right)=\lambda_1.$
\end{proof}

\begin{lemma}\label{AC}
For $\hat{R}^{X}_{T} = \frac{(1-2(1-\alpha)p)\left(\SX\right)^{2}}{2\left(1-(1-\alpha)p\right)^{2}}T +
\frac{\SX}{1-(1-\alpha)p} W^{X}_{T}$,

\noindent when $p=0$,
\begin{equation*}
\lim_{T\rightarrow\infty}\frac{1}{T}\mathbb{E}\left[\ln\left(e^{(1-\alpha)\bar{\hat{R}}^{X}_{T}}-1\right)\right] = \lim_{T\rightarrow\infty}\frac{1}{T}\mathbb{E}\left[(1-\alpha)\bar{\hat{R}}^{X}_{T}\right],
\end{equation*}
and when $0<p<1$,
\begin{equation*}
\lim_{T\rightarrow\infty}\frac{1}{T}\ln\mathbb{E}\left[\left(e^{(1-\alpha)\bar{\hat{R}}^{X}_{T}}-1\right)^{p}\right]^{\frac{1}{p}} = \lim_{T\rightarrow\infty}\frac{1}{T}\ln\mathbb{E}\left[e^{(1-\alpha)p\bar{\hat{R}}^{X}_{T}}\right]^{\frac{1}{p}}.
\end{equation*}
\end{lemma}

\begin{proof}
Since $e^{(1-\alpha)\bar{\hat{R}}^{X}_{T}}-1\leq e^{(1-\alpha)\bar{\hat{R}}^{X}_{T}}$, the $``\leq"$ part of the two equations is straight forward. It suffices to prove the $``\geq"$ part.

When $p=1$, since $e^{(1-\alpha)\bar{\hat{R}}^{X}_{T}} = e^{(1-\alpha)\bar{\hat{R}}^{X}_{T}} -1 +1$, the concavity of logarithmic function implies that
\begin{align}
&(1-\alpha)\bar{\hat{R}}^{X}_{T} = \ln e^{(1-\alpha)\bar{\hat{R}}^{X}_{T}} \leq \ln\left(e^{(1-\alpha)\bar{\hat{R}}^{X}_{T}}-1\right) + \frac{1}{e^{(1-\alpha)\bar{\hat{R}}^{X}_{T}}-1}\label{con11}
\end{align}

Since $\SX >0$, $e^{(1-\alpha)\hat{R}^{X}_{T}} = e^{\frac{(1-\alpha)\left(\SX\right)^{2}}{2}T + (1-\alpha)\SX W^{X}_{T}} \geq 2e^{(1-\alpha)\SX W^{X}_{T}}$ for sufficiently large $T$, and
\begin{equation}
e^{(1-\alpha)\bar{\hat{R}}^{X}_{T}} \geq 2e^{(1-\alpha)\SX \bar W^{X}_{T}} \geq 1 + e^{(1-\alpha)\SX \bar W^{X}_{T}}.\label{cont}
\end{equation}

Thus,
\begin{equation}
\lim_{T\rightarrow\infty}\frac{1}{T}\mathbb{E}\left[\frac{1}{e^{(1-\alpha)\bar{\hat{R}}^{X}_{T}}-1}\right]  \leq \lim_{T\rightarrow\infty}\frac{1}{T}\mathbb{E}\left[e^{-(1-\alpha)\SX \bar W^{X}_{T}}\right]=0.\label{con12}
\end{equation}

Then, (\ref{con11}) and (\ref{con12}) imply taht
\begin{align*}
&\lim_{T\rightarrow\infty}\frac{1}{T}\mathbb{E}\left[(1-\alpha)\bar{\hat{R}}^{X}_{T} \right] \leq \lim_{T\rightarrow\infty}\frac{1}{T}\mathbb{E}\left[\ln\left(e^{(1-\alpha)\bar{\hat{R}}^{X}_{T}} -1\right)\right] +\lim_{T\rightarrow\infty}\frac{1}{T}\mathbb{E}\left[\frac{1}{e^{(1-\alpha)\bar{\hat{R}}^{X}_{T}}-1}\right]\\
\leq& \lim_{T\rightarrow\infty}\frac{1}{T}\mathbb{E}\left[\ln\left(e^{(1-\alpha)\bar{\hat{R}}^{X}_{T}} -1\right)\right].
\end{align*}

When $0<p<1$, from the concavity of the function $f(x) = x^{p}$,
\begin{equation*}
e^{(1-\alpha)p\bar{\hat{R}}^{X}_{T}} \leq \left(e^{(1-\alpha)\bar{\hat{R}}^{X}_{T}}-1\right)^{p} + p\left(e^{(1-\alpha)\bar{\hat{R}}^{X}_{T}}-1\right)^{p-1}.
\end{equation*}

Thus, from Lemma 1.2.15 in \citet*{Dembolargedeviation},
\begin{align}
&\lim_{T\rightarrow\infty}\frac{1}{T}\ln\mathbb{E}\left[(1-\alpha)p\bar{\hat{R}}^{X}_{T} \right]^{\frac{1}{p}}\nonumber\\
\leq& \lim_{T\rightarrow\infty}\frac{1}{T}\ln\mathbb{E}\left[\left(e^{(1-\alpha)\bar{\hat{R}}^{X}_{T}}-1\right)^{p} + p\left(e^{(1-\alpha)\bar{\hat{R}}^{X}_{T}}-1\right)^{p-1}\right]^{\frac{1}{p}}\nonumber\\
=& \max\left(\lim_{T\rightarrow\infty}\frac{1}{T}\ln\mathbb{E}\left[\left(e^{(1-\alpha)\bar{\hat{R}}^{X}_{T}}-1\right)^{p}\right]^{\frac{1}{p}}, \lim_{T\rightarrow\infty}\frac{1}{T}\ln\mathbb{E}\left[p\left(e^{(1-\alpha)\bar{\hat{R}}^{X}_{T}}-1\right)^{p-1}\right]^{\frac{1}{p}}\right).\label{con21}
\end{align}

Since $p-1<0$, from (\ref{cont}),
\begin{equation}
\lim_{T\rightarrow\infty}\frac{1}{T}\ln\mathbb{E}\left[p\left(e^{(1-\alpha)\bar{\hat{R}}^{X}_{T}}-1\right)^{p-1}\right]^{\frac{1}{p}} \leq \lim_{T\rightarrow\infty}\frac{1}{T}\ln\mathbb{E}\left[pe^{(1-\alpha)(p-1)\SX \bar W^{X}_{T}}\right]^{\frac{1}{p}}\leq 0.\label{con22}
\end{equation}

Thus, (\ref{con21}) and (\ref{con22}) imply that
\begin{align*}
&\lim_{T\rightarrow\infty}\frac{1}{T}\ln\mathbb{E}\left[e^{(1-\alpha)p\bar{\hat{R}}^{X}_{T}} \right]^{\frac{1}{p}}\leq \max\left(\lim_{T\rightarrow\infty}\frac{1}{T}\ln\mathbb{E}\left[\left(e^{(1-\alpha)\bar{\hat{R}}^{X}_{T}}-1\right)^{p}\right]^{\frac{1}{p}}, 0\right)\\
\leq& \lim_{T\rightarrow\infty}\frac{1}{T}\ln\mathbb{E}\left[\left(e^{(1-\alpha)\bar{\hat{R}}^{X}_{T}} -1\right)^{p}\right]^{\frac{1}{p}},
\end{align*}
because from (\ref{cont}),
\begin{align*}
&\lim_{T\rightarrow\infty}\frac{1}{T}\ln\mathbb{E}\left[\left(e^{(1-\alpha)\bar{\hat{R}}^{X}_{T}} -1\right)^{p}\right]^{\frac{1}{p}}\\
\geq& \lim_{T\rightarrow\infty}\frac{1}{T}\ln\mathbb{E}\left[e^{(1-\alpha)p\SX \bar W^{X}_{T}}\right]^{\frac{1}{p}} = \frac{(1-\alpha)^2\left(\SX\right)^{2}p}{2}>0.
\end{align*}
\end{proof}

\subsection{Proof of Theorem \ref{Mutual}}
Solving the differential equations (\ref{DymXM}) and (\ref{DymFM}),
\begin{align*}
X_t =& X_0e^{R^X_t - \varphi t},\\
F_t =& F_0e^{R^F_t} + \int_0^t \varphi e^{R^F_{s,t}}X_sds.
\end{align*}

Theorem \ref{Mutual} is proved by arguments similar to those for Theorem \ref{Main}: first prove that for general strategies the ESR of wealth is bounded above by (\ref{ESRM}), and then show that the ESR induced by the candidate strategies $\optpix$ and $\optpif$ in (\ref{optimal1M}) and (\ref{optimal2M}) achieve this upper bound.

\begin{lemma}\label{lessM}
For a mutual fund manager compensated by proportional fees with rate $\varphi >0$, the $\ESR$ induced by any investment strategies $\pix$ and $\pif$ satisfies
\begin{equation*}
\ESR_{\gamma}(\pif,\pif) \leq \lambda_2 = \max\left(\frac{\left(\SX\right)^{2}}{2\gamma}-\varphi,\frac{\left(\SF\right)^2}{2\gamma}\right), \text{ for any } 0<\gamma\leq 1.
\end{equation*}
\end{lemma}

\begin{proof}
We prove this lemma for logarithmic utility and power utility, respectively.

\noindent\textit{The case of logarithmic utility}.
\begin{equation*}
\lim_{T\rightarrow \infty} \frac{1}{T} \mathbb{E}\left[\ln F_{T}\right] = \lambda_2 + \lim_{T\rightarrow \infty} \frac{1}{T} \mathbb{E}\left[\ln \left(F_0e^{R^F_T-\lambda_2 T} + \varphi\int_0^Te^{R^F_{t,T}-\lambda_2 T}X_tdt\right)\right].
\end{equation*}

Then (\ref{LN}) implies that, with $N^T_t$ defined in the proof of Lemma \ref{ADDWF}
\small
\begin{align}
&\lim_{T\rightarrow \infty} \frac{1}{T} \mathbb{E}\left[\ln \left(F_0e^{R^F_T-\lambda_2 T} + \varphi\int_0^Te^{R^F_{t,T}-\lambda_2 T}X_tdt\right)\right]\nonumber\\
=&\lim_{T\rightarrow \infty} \frac{1}{T} \mathbb{E}\left[\ln \left(F_0e^{R^F_T-\lambda_2 T} + \varphi\int_0^Te^{R^F_{t,T}-\lambda_2 T}X_tdt\right)\right] + \lim_{T \rightarrow \infty}\frac{1}{T}\mathbb{E}\left[-\SF\bar N^{T}_{T}-\SX\overline{\rho W^{F}_{T}}\right]\nonumber\\
\leq&\lim_{T\rightarrow \infty} \frac{1}{T} \mathbb{E}\left[\ln \left(F_0e^{R^F_T- \frac{\left(\SF\right)^2T}{2} - \SF W^F_T} + \varphi\int_0^Te^{R^F_{t,T}- \frac{\left(\SF\right)^2(T-t)}{2}-\SF W^F_{t,T}}e^{-\left(\frac{\left(\SX\right)^2}{2}-\varphi\right)t - \rho \SX W^F_t}X_tdt\right)\right]\label{ADDM},
\end{align}
\normalsize
where the last inequality follows from the definition of $\lambda_2$ and the fact that $\bar N^T_T \geq W_{t,T}$ for any $0\leq t\leq T$. Then with $W^{\perp}$, the Brownian Motion independent to $W^F$ and that $W^X_t = \rho W^F_t + \sqrt{1-\rho^2}W^{\perp}_t$, Jensen's inequality implies that (\ref{ADDM}) is less than or equal to
\small
\begin{equation}
\lim_{T\rightarrow \infty} \frac{1}{T} \mathbb{E}\left[\ln \mathbb{E}_{W^{\perp}_T}\left[F_0e^{R^F_T- \frac{\left(\SF\right)^2T}{2} - \SF W^F_T} + \varphi\int_0^Te^{R^F_{t,T}- \frac{\left(\SF\right)^2(T-t)}{2}-\SF W^F_{t,T}}e^{-\left(\frac{\left(\SX\right)^2}{2}-\varphi\right)t - \rho \SX W^F_t}X_tdt\right]\right]\label{BSM}.
\end{equation}
\normalsize

From Lemma \ref{SuperM}, $M_t = e^{R^F_t- \frac{\left(\SF\right)^2t}{2} - \SF W^F_t}$ is a super-martingale with respect to the filtration generated by $\left(W^{\perp}_s\right)_{0\leq s\leq T}$ and $\left(W^{F}_s\right)_{0\leq s\leq t}$. Then, Lemma \ref{IntCE} implies that (\ref{BSM}) is less than or equal to
\begin{align}
&\lim_{T\rightarrow \infty} \frac{1}{T} \mathbb{E}\left[\ln \mathbb{E}_{W^{\perp}_T}\left[F_0+ \varphi\int_0^Te^{-\left(\frac{\left(\SX\right)^2}{2}-\varphi\right)t - \rho \SX W^F_t}X_tdt\right]\right]\nonumber\\
=&\lim_{T\rightarrow \infty} \frac{1}{T} \mathbb{E}\left[\ln \mathbb{E}_{W^{\perp}_T}\left[F_0+ \varphi\int_0^Te^{-\left(\frac{\left(\SX\right)^2}{2}-\varphi\right)t - \rho \SX W^F_t}X_tdt\right]\right]\nonumber\\
+& \lim_{T\rightarrow \infty} \frac{1}{T} \mathbb{E}\left[\ln e^{-\sqrt{1-\rho^2}\SX \bar W^{\perp}_T}\right],\label{BAWP}
\end{align}
where the last equality follows from (\ref{LN}). Then, since $\bar W^{\perp}_T \geq W^{\perp}_t$ for every $0\leq t\leq T$ and $X_t = X_0e^{R^X_t-\varphi t}$, (\ref{BAWP}) is less than or equal to
\begin{align*}
\textstyle
&\lim_{T\rightarrow \infty} \frac{1}{T} \mathbb{E}\left[\ln \mathbb{E}_{W^{\perp}_T}\left[F_0 e^{-\sqrt{1-\rho^2}\SX \bar W^{\perp}_T}+ \varphi X_0\int_0^Te^{R^X_t-\frac{\left(\SX\right)^2 t}{2} - \rho\SX W^F_t -\sqrt{1-\rho^2}\SX W^{\perp}_t}dt\right]\right]\\
=&\lim_{T\rightarrow \infty} \frac{1}{T} \mathbb{E}\left[\ln \mathbb{E}_{W^{\perp}_T}\left[F_0 e^{-\sqrt{1-\rho^2}\SX \bar W^{\perp}_T}+ \varphi X_0\int_0^Te^{R^X_t-\frac{\left(\SX\right)^2 t}{2} - \SX W^X_t}dt\right]\right]\\
\leq&\lim_{T\rightarrow \infty} \frac{1}{T}\ln \mathbb{E}\left[F_0 e^{-\sqrt{1-\rho^2}\SX \bar W^{\perp}_T}+ \varphi X_0\int_0^Te^{R^X_t-\frac{\left(\SX\right)^2 t}{2} - \SX W^X_t}dt\right],
\end{align*}
\normalsize
where the last inequality follows from Jensen's inequality and the tower property of conditional expectation. Then since $G_t = e^{R^X_t-\frac{\left(\SX\right)^2 t}{2} - \SX W^X_t}$ is a super-martingale, by Fubini's Theorem, the above equals to
\begin{align*}
&\lim_{T\rightarrow \infty} \frac{1}{T}\ln \left(\mathbb{E}\left[F_0 e^{-\sqrt{1-\rho^2}\SX \bar W^{\perp}_T}\right]+ \varphi X_0\int_0^T\mathbb{E}\left[e^{R^X_t-\frac{\left(\SX\right)^2 t}{2} - \SX W^X_t}\right]dt\right)\\
\leq&\lim_{T\rightarrow \infty} \frac{1}{T}\ln \left(F_0 + \varphi X_0 \int_0^T1dt\right)= \lim_{T\rightarrow \infty} \frac{1}{T}\ln \left(F_0 + \varphi X_0T\right)\leq 0,
\end{align*}
and this implies that $\displaystyle\lim_{T\rightarrow \infty} \frac{1}{T} \mathbb{E}\left[\ln F_{T}\right]\leq \lambda_2$.\\

\noindent\textit{The case of power utility}. Define $\tilde{F}_t = F_t + \bar X_t$, which implies that $\tilde{F}_t \geq \bar X_t$, $\tilde{F}_t \geq F_t$. Thus the ESR of $F$ is less than or equal to the ESR of $\tilde{F}$, which will be proved in the following to be also bounded above by $\lambda_2$. Notice that this is similar to the technique used to deal with the power utility case in the proof of Lemma \ref{less}: adding a positive and increasing process to wealth process without increasing the ESR. Here we choose $\bar X$, because the property $\bar X_t\geq X_t$ helps to derive \eqref{ctx} below, though the mutual fund manager is not compensated by high-water mark performance fees.

Let $\tilde{\pi}^F_t = \frac{\tilde{F}_t - \bar X_t}{\tilde{F}_t}\pif_t$, then for $0<p<1$,
\begin{align*}
d\tilde{F}^p_t =& p\tilde{F}^{p-1}_t\left(\tilde{F}_t -  \bar X_t\right)\left(\pif_t\muf dt+\pif_t\sigmaf dW^F_t\right) \\
+& \frac{p(p-1)}{2}\tilde{F}^{p-2}_t\left(\tilde{F}_t -  \bar X_t\right)^2\left(\pif_t\sigmaf\right)^2dt + \varphi p\tilde{F}^{p-1}_t X_tdt + p\tilde{F}^{p-1} d\bar X_t\\
=&p\tilde{F}^p\left(\left(\tilde{\pi}^F_t\muf + \frac{p-1}{2}\left(\tilde{\pi}^F_t\sigmaf\right)^2\right)dt + \tilde{\pi}^F_t\sigmaf dW^F_t\right)\\
 +& \varphi p\tilde{F}^{p-1}_t X_tdt +  p\tilde{F}^{p-1} d\bar X_t.
\end{align*}

Solving this differential equation, $\tilde{F}^p_t$ can be represented as a sum of three positive processes,
\begin{equation}
\tilde{F}^p_T = F_0^pe^{pR^{F,\tilde{\pi}^F}_T} + \varphi p\int_0^Te^{pR^{F,\tilde{\pi}^F}_{t,T}}\tilde{F}^{p-1}_t X_tdt + p\int_0^Te^{pR^{F,\tilde{\pi}^F}_{t,T}}\tilde{F}^{p-1}d\bar X_t.\label{three}
\end{equation}

Then, from \citet*{Dembolargedeviation}, Lemma $1.2.15$, if suffices to prove that the ESR of each of the three terms in (\ref{three}) is less than or equal to $\lambda_2$.

From (\ref{max1}), $\displaystyle\lim_{T\rightarrow \infty} \frac{1}{T} \ln\mathbb{E}\left[F_0^pe^{pR^{F,\tilde{\pi}^F}_T} \right]^{\frac{1}{p}} \leq \frac{\left(\SF\right)^2}{2(1-p)}\leq \lambda_2$.

For the second term in (\ref{three}), since $\tilde{F}_t\geq \bar X_{t} \geq X_t$ and $p-1<0$,
\begin{align}
&\lim_{T\rightarrow \infty} \frac{1}{T} \ln\mathbb{E}\left[\varphi p\int_0^Te^{pR^{F,\tilde{\pi}^F}_{t,T}}\tilde{F}^{p-1}_t X_tdt\right]^{\frac{1}{p}} \leq \lim_{T\rightarrow \infty} \frac{1}{T} \ln\mathbb{E}\left[\int_0^Te^{pR^{F,\tilde{\pi}^F}_{t,T}} X^{p}_tdt\right]^{\frac{1}{p}}\label{ctx}\\
=&\lambda_2 + \lim_{T\rightarrow \infty} \frac{1}{T} \ln\mathbb{E}\left[\int_0^Te^{pR^{F,\tilde{\pi}^F}_{t,T}-p\lambda_2 (T-t)} X^{p}_te^{-p\lambda_2 t}dt\right]^{\frac{1}{p}}\nonumber\\
=& \lambda_2 + \lim_{T\rightarrow \infty} \frac{1}{pT} \ln \int_0^T\mathbb{E}\left[e^{pR^{F,\tilde{\pi}^F}_{t,T}-p\lambda_2 (T-t)} X^{p}_te^{-p\lambda_2 t}\right] dt,\nonumber
\end{align}
where the last equality follows from Fubini's Theorem. Then, since (\ref{WEXP}) implies that $\mathbb{E}_t\left[e^{pR^{F,\tilde{\pi}^{F}}_{t,T}}\right] \leq  e^{\frac{p\left(\SF\right)^{2}}{2(1-p)}(T-t)}\leq e^{p\lambda_2(T-t)}$, from the tower property of conditional expectation,
\begin{align*}
&\lim_{T\rightarrow \infty} \frac{1}{pT} \ln \int_0^T\mathbb{E}\left[e^{pR^{F,\tilde{\pi}^F}_{t,T}-p\lambda_2 (T-t)} X^{p}_te^{-p\lambda t}\right] dt\\
=&\lim_{T\rightarrow \infty} \frac{1}{pT} \ln \int_0^T\mathbb{E}\left[\mathbb{E}_t\left[e^{pR^{F,\tilde{\pi}^F}_{t,T}-p\lambda_2 (T-t)}\right] X^{p}_te^{-p\lambda_2 t}\right] dt\\
\leq&\lim_{T\rightarrow \infty} \frac{1}{pT} \ln \int_0^T\mathbb{E}\left[X^{p}_te^{-p\lambda_2 t}\right] dt.
\end{align*}

Since $\mathbb{E}\left[X^{p}_t\right] = \mathbb{E}\left[X^p_0e^{pR^X_t - p\varphi t}\right] \leq X^p_0e^{p\left(\frac{\left(\SX\right)^2}{2(1-p)}-\varphi\right)t}$, which follows from the argument similar to (\ref{WEXP}), and $\lambda_2\geq \frac{\left(\SX\right)^2}{2(1-p)}-\varphi$, $\displaystyle\lim_{T\rightarrow \infty} \frac{1}{pT} \ln \int_0^T\mathbb{E}\left[X^{p}_te^{-p\lambda_2 t}\right]dt \leq \displaystyle\lim_{T\rightarrow \infty} \frac{1}{pT}\ln \left(X_0^pT\right) = 0$, which implies that $\displaystyle\lim_{T\rightarrow \infty} \frac{1}{T} \ln\mathbb{E}\left[\varphi p\int_0^Te^{pR^{F,\tilde{\pi}^F}_{t,T}}\tilde{F}^{p-1}_t X_tdt\right]^{\frac{1}{p}}\leq \lambda_2$.

Finally, since $\tilde{F}_t\geq \bar X_{t}$ and $p-1<0$, following arguments similar to those in the proof of Lemma \ref{UBofF},
\begin{align*}
&\lim_{T\rightarrow \infty} \frac{1}{T} \ln\mathbb{E}\left[p\int_0^Te^{pR^{F,\tilde{\pi}^F}_{t,T}}\tilde{F}^{p-1}d\bar X_t\right]^{\frac{1}{p}} \leq \lim_{T\rightarrow \infty} \frac{1}{pT} \ln\mathbb{E}\left[p\int_0^Te^{pR^{F,\tilde{\pi}^F}_{t,T}}\bar X_t^{p-1}d\bar X_t\right]\\
=&\lim_{T\rightarrow \infty} \frac{1}{pT} \ln\mathbb{E}\left[\int_0^Te^{pR^{F,\tilde{\pi}^F}_{t,T}}d\bar X_t^p\right] \leq \lambda_2 + \lim_{T\rightarrow \infty} \frac{1}{pT} \ln\mathbb{E}\left[\int_0^Te^{-p\lambda_2 t}d\bar X_t^p\right].
\end{align*}

By integration by parts, $\int_0^Te^{-p\lambda_2 t}d\bar X_t^p \leq (1+p\lambda_2 T)X^p_0e^{p\left(\overline{R^X_{\cdot}-\varphi\cdot-\lambda_2\cdot}\right)_T}$. Then, applying Lemma 9 in \citet*{guasoni2012nopart} to the case of $\alpha = 0$ and $ r = -\lambda_2$ implies that
\begin{align*}
&\lim_{T\rightarrow \infty} \frac{1}{pT} \ln\mathbb{E}\left[\int_0^Te^{-p\lambda_2 t}d\bar X_t^p\right]\\
 \leq& \lim_{T\rightarrow \infty} \frac{1}{pT} \ln\mathbb{E}\left[e^{p\left(\overline{R^X_{\cdot}-\varphi\cdot-\lambda_2\cdot}\right)_T}\right]
\leq  \frac{\left(\SX\right)^2}{2(1-p)} -\varphi -\lambda_2\leq 0,
\end{align*}
which indicates that $\displaystyle\lim_{T\rightarrow \infty} \frac{1}{T} \ln\mathbb{E}\left[p\int_0^Te^{pR^{F,\tilde{\pi}^F}_{t,T}}\tilde{F}^{p-1}d\bar X_t\right]^{\frac{1}{p}}\leq \lambda_2$.
\end{proof}

Now it remains to prove that the ESR induced by the candidate strategies in (\ref{optimal1M}) and (\ref{optimal2M}) achieves (\ref{ESRM}).

\begin{lemma}\label{optimumM}
For $\optpix$ and $\optpif$ in (\ref{optimal1M}) and (\ref{optimal2M}), $\ESR_{\gamma}\left(\hat{\pi}^{X},\hat{\pi}^{F}\right) = \lambda_2$, for any $0<\gamma\leq 1$.
\end{lemma}

\begin{proof}
Plugging $\optpix$ and $\optpif$ into (\ref{DymXM}) and (\ref{DymFM}), with $\hat{R}^X = \frac{(1-2p)\left(\SX\right)^2T}{2(1-p)^2} + \frac{\SX}{1-p}W^X_T$ and $\hat{R}^F = \frac{(1-2p)\left(\SF\right)^2T}{2(1-p)^2} + \frac{\SF}{1-p}W^F_T$, the corresponding fund's value and wealth satisfy:
\begin{align*}
\hat{X}_T =& X_0e^{\hat{R}^X_T-\varphi T},\\
\hat{F}_T =& F_0e^{\hat{R}^F_T} + \varphi\int_0^T\hat{X}_tdt.
\end{align*}

Let $G_{t} = F_{0}e^{\hat{R}^{F}_{T}}$ and $H_{t} =\varphi\int_0^T\hat{X}_tdt$, then $\hat{F}_{T} = G_{T} + H_{T}$. From Lemma \ref{lessM}, it suffices to prove that $\ESR_{\gamma}\left(\optpix,\optpif\right) \geq \lambda_2$.\\

\noindent \textit{The case of logarithmic utility}. Since $H_{t}$ is a positive process,
\begin{equation*}
\lim_{T\rightarrow\infty}\frac{1}{T}\mathbb{E}\left[\ln\hat{F}_{T}\right]=\lim_{T\rightarrow\infty}\frac{1}{T}\mathbb{E}\left[\ln(G_{T} + H_{T})\right]\geq\lim_{T\rightarrow\infty}\frac{1}{T}\mathbb{E}\left[\ln G_{T}\right]=\frac{\left(\SF\right)^{2}}{2}.
\end{equation*}

Likewise, since $G_{t}$ is a positive process,
\begin{align}
&\lim_{T\rightarrow\infty}\frac{1}{T}\mathbb{E}\left[\ln\hat{F}_{T}\right]\geq\lim_{T\rightarrow\infty}\frac{1}{T}\mathbb{E}\left[\ln H_{T}\right]=\lim_{T\rightarrow\infty}\frac{1}{T}\mathbb{E}\left[\ln \left(\varphi\int_0^T\hat{X}_tdt\right)\right]\nonumber\\
=& \lim_{T\rightarrow \infty} \frac{1}{T} \mathbb{E}\left[\ln \int_0^Te^{\left(\frac{\left(\SX\right)^2}{2}-\varphi\right)t + \SX W^X_t}dt \right]\nonumber\\
\geq& \lim_{T\rightarrow \infty} \frac{1}{T} \mathbb{E}\left[\ln \int_0^Te^{\left(\frac{\left(\SX\right)^2}{2}-\varphi\right)t + \SX \underbar W^X_{t}}dt \right]\label{minW}.
\end{align}
Then since $\underbar W^X_{t}\geq \underbar W^X_{T}$ for every $0\leq t\leq T$, \eqref{minW} is greater than or equal to
\begin{align*}
& \lim_{T\rightarrow \infty} \frac{1}{T} \mathbb{E}\left[\SX \underbar W^X_{T} + \ln \int_0^Te^{\left(\frac{\left(\SX\right)^2}{2}-\varphi\right)t} dt \right] = \lim_{T\rightarrow \infty} \frac{1}{T}\ln \int_0^Te^{\left(\frac{\left(\SX\right)^2}{2}-\varphi\right)t} dt\\
=&\lim_{T\rightarrow \infty} \frac{1}{T}\ln \frac{e^{\left(\frac{\left(\SX\right)^2}{2}-\varphi\right)T}-1}{\frac{\left(\SX\right)^2}{2}-\varphi} = \max\left(\frac{\left(\SX\right)^2}{2}-\varphi,0\right),
\end{align*}
where the first equality follows from (\ref{LN}).

Thus, $\displaystyle\lim_{T\rightarrow\infty}\frac{1}{T}\mathbb{E}\left[\ln\hat{F}_{T}\right] \geq \max\left(\frac{\left(\SF\right)^{2}}{2},  \max\left(\frac{\left(\SX\right)^2}{2}-\varphi,0\right)\right) = \lambda_2$.

\noindent \textit{The case of power utility}. Since $H_{t}$ is a positive process,
\begin{equation*}
\lim_{T\rightarrow\infty}\frac{1}{T}\ln\mathbb{E}\left[\hat{F}^p_{T}\right]^{\frac{1}{p}}\geq\lim_{T\rightarrow\infty}\frac{1}{T}\ln\mathbb{E}\left[G^p_{T}\right]^{\frac{1}{p}}=\frac{\left(\SF\right)^{2}}{2(1-p)}.
\end{equation*}

Likewise, since $G_{t}$ is a positive process,
\begin{align*}
&\lim_{T\rightarrow\infty}\frac{1}{T}\ln\mathbb{E}\left[\hat{F}^p_{T}\right]^{\frac{1}{p}}\geq\lim_{T\rightarrow\infty}\frac{1}{T}\ln \mathbb{E}\left[ H^p_{T}\right]^{\frac{1}{p}}=\lim_{T\rightarrow\infty}\frac{1}{T}\ln\mathbb{E}\left[\varphi^p\left(\int_0^T\hat{X}_tdt\right)^p\right]^{\frac{1}{p}}\\
=&\lim_{T\rightarrow \infty} \frac{1}{pT} \ln \mathbb{E}\left[\left(\int_0^Te^{\left(\frac{(1-2p)\left(\SX\right)^2}{2(1-p)^2}-\varphi\right)t + \frac{\SX}{1-p} W^X_t}dt \right)^p\right] \\
=& \lim_{T\rightarrow \infty} \frac{1}{pT} \ln \mathbb{E}\left[T^p\left(\int_0^T\frac{1}{T}e^{\left(\frac{(1-2p)\left(\SX\right)^2}{2(1-p)^2}-\varphi\right)t + \frac{\SX}{1-p} W^X_t}dt \right)^p\right].
\end{align*}

Since $0<p<1$, Jensen's inequality implies that
\begin{align*}
&\mathbb{E}\left[T^p\left(\int_0^T\frac{1}{T}e^{\left(\frac{(1-2p)\left(\SX\right)^2}{2(1-p)^2}-\varphi\right)t + \frac{\SX}{1-p} W^X_t}dt \right)^p\right]\\
 \geq& \mathbb{E}\left[T^{p-1}\int_0^Te^{\left(\frac{p(1-2p)\left(\SX\right)^2}{2(1-p)^2}-p\varphi\right)t + \frac{p\SX}{1-p} W^X_t}dt\right]\\
=&T^{p-1}\int_0^T\mathbb{E}\left[e^{\left(\frac{p(1-2p)\left(\SX\right)^2}{2(1-p)^2}-p\varphi\right)t + \frac{p\SX}{1-p} W^X_t}\right]dt = T^{p-1}\int_0^T e^{\left(\frac{p\left(\SX\right)^2}{2(1-p)}-p\varphi\right)t}dt,
\end{align*}
where the first equality follows from Fubini's Theorem. Thus,
\begin{align*}
&\lim_{T\rightarrow\infty}\frac{1}{T}\ln\mathbb{E}\left[\hat{F}^p_{T}\right]^{\frac{1}{p}}
\geq\lim_{T\rightarrow \infty} \frac{1}{pT} \ln \mathbb{E}\left[T^p\left(\int_0^T\frac{1}{T}e^{\left(\frac{(1-2p)\left(\SX\right)^2}{2(1-p)^2}-\varphi\right)t + \frac{\SX}{1-p} W^X_t}dt \right)^p\right]\\
\geq&\lim_{T\rightarrow \infty} \frac{1}{pT} \ln \left(T^{p-1}\int_0^T e^{\left(\frac{p\left(\SX\right)^2}{2(1-p)}-p\varphi\right)t}dt\right) = \max\left(\frac{\left(\SX\right)^2}{2(1-p)}-\varphi,0\right),
\end{align*}
which implies that $\displaystyle\lim_{T\rightarrow\infty}\frac{1}{T}\ln\mathbb{E}\left[\hat{F}^p_{T}\right]^{\frac{1}{p}} \geq \max\left(\frac{\left(\SF\right)^{2}}{2(1-p)}, \max\left(\frac{\left(\SX\right)^2}{2(1-p)}-\varphi,0\right)\right) = \lambda_2$.

\end{proof}

\bibliographystyle{agsm}
\bibliography{watermark}
\end{document}